\documentclass[a4paper,UKenglish,cleveref,autoref,thm-restate]{article}

\usepackage{microtype}
\usepackage{hyperref}
\usepackage[usenames,dvipsnames]{xcolor}
\usepackage{algorithm}
\usepackage{algorithmicx}
\usepackage[noend]{algpseudocode}
\usepackage[noadjust]{cite}
\usepackage[OT4]{fontenc}
\usepackage{amsthm}
\usepackage{amssymb}
\usepackage{amsmath}
\usepackage{amsfonts}
\usepackage{todonotes}
\usepackage{rotating}
\usepackage{xspace}
\usepackage{verbatim}
\usepackage{mathrsfs}
\usepackage{tablefootnote}
\usepackage{multirow}
\usepackage{authblk}
\usepackage[margin=1in]{geometry}
\newcommand{\Oh}{\mathcal{O}}

\newcommand{\eps}{\varepsilon}
\newcommand{\elbow}{elbow}
\newcommand{\NULL}{\texttt{NULL}}
\newcommand{\this}{\textsf{this}}
\newcommand{\False}{\textsf{false}}
\newcommand{\True}{\textsf{true}}
\newcommand{\Next}{\textsc{Next}}
\newcommand{\Ptr}{\textsc{Ptr}}
\newcommand{\Advance}{\textsc{Advance}}
\newcommand{\Process}{\textsc{Process}}
\newcommand{\BestStaircase}{\textsc{BestStaircase}}
\newcommand{\BestBStaircase}{\textsc{Best}\textit{b}\textsc{Staircase}}

\newcommand{\Min}{\textsc{MinRange}}
\newcommand{\InvertCycle}{\textsc{InvertCycle}}
\newcommand{\PathEnd}{\textsc{PathEnd}}
\newcommand{\FindIntersection}{\textsc{FindIntersection}}

\newcommand{\GetEnd}{\textsc{GetEnd}}

\newcommand{\myarrow}[1]{\overset{\pi_{#1}}{\longrightarrow}}
\newcommand{\len}[2]{\ell(#1,#2)}
\newcommand{\rin}[1]{\rotatebox[origin=c]{75}{$\in #1$}}
\newcommand{\false}{\textsc{false}}
\newcommand{\true}{\textsc{true}}
\newcommand{\Try}{\textbf{try}}
\newcommand{\Catch}{\textbf{catch}}

\newtheorem{theorem}{Theorem}[section]

\newtheorem{corollary}[theorem]{Corollary}
\newtheorem{lemma}[theorem]{Lemma}
\newtheorem{proposition}[theorem]{Proposition}
\newtheorem{observation}[theorem]{Observation}

\theoremstyle{definition}   
\newtheorem{definition}[theorem]{Definition}

\usepackage{etoolbox}
\makeatletter
\setbool{@fleqn}{false}
\makeatother
\renewcommand{\paragraph}{\subparagraph}


\title{Strictly In-Place Algorithms for Permuting and Inverting Permutations}
\author[1]{Bart{\l}omiej Dudek}
\author[1]{Pawe{\l} Gawrychowski}
\author[1]{Karol Pokorski}
\affil[1]{Institute of Computer Science, University of Wroc{\l}aw, Poland
[bartlomiej.dudek,gawry,pokorski]@cs.uni.wroc.pl
}
\date{}

\newcommand{\FIGURE}[4]{
\begin{figure}[#1]
\begin{centering}
\includegraphics[width=#2\textwidth]{figures/#3.pdf}
\caption{#4}
\label{fig:#3}
\end{centering}
\end{figure}
}

\begin{document}
\maketitle

\begin{abstract}
We revisit the problem of permuting an array of length $n$ according to a given permutation
in place, that is, using only a small number of bits of extra storage. Fich, Munro and Poblete
[FOCS 1990, SICOMP 1995] obtained an elegant $\Oh(n\log n)$-time algorithm using only $\Oh(\log^{2}n)$ bits of extra
space for this basic problem by designing a procedure that scans the permutation and outputs
exactly one element from each of its cycles. However, in the strict sense in place should be understood
as using only an asymptotically optimal $\Oh(\log n)$ bits of extra space, or storing a constant number
of indices. The problem of permuting in this version is, in fact, a well-known interview question, with the
expected solution being a quadratic-time algorithm. Surprisingly, no faster algorithm seems to be known
in the literature.

Our first contribution is a strictly in-place generalisation of the method of Fich et al. that works in $\Oh_{\eps}(n^{1+\eps})$
time, for any $\eps > 0$. Then, we build on this generalisation to obtain a strictly in-place algorithm for
inverting a given permutation on $n$ elements working in the same complexity.
This is a significant improvement on a recent result of Gu\'spiel [arXiv 2019], who designed an $\Oh(n^{1.5})$-time algorithm.
\end{abstract}

\section{Introduction}
 
Permutations are often used as building blocks in combinatorial algorithms
operating on more complex objects. This brings the need for being able to efficiently operate on permutations.
One of the most fundamental operations is rearranging an array $A[1..n]$ according to a permutation $\pi$.
This can be used, for example, to transpose a rectangular array~\cite[Ex.~1.3.3-12]{KnuthVol1}.
Denoting by $a_{i}$ the value stored in $A[i]$, the goal is to make every $A[i] = a_{\pi^{-1}(i)}$.
This is trivial if we can allocate a temporary array $B[1..n]$
and, after setting $B[i] \gets A[\pi(i)]$ for every $i$, copy $B[1..n]$ to $A[1..n]$. Alternatively, one can
iterate over the cycles of $\pi$ and rearrange the values on each cycle.
Then there is no need for allocating a temporary array as long as we can recognise
the elements of $\pi$ in already processed cycles. This is easy if
we can overwrite $\pi$, say by setting $\pi(i) \gets i$ after having processed $\pi(i)$.
However, we might want to use the same $\pi$ later,
and thus cannot overwrite its elements. In such a case, assuming that every $\pi(i)$
can store at least one extra bit, we could mark the
processed elements by temporarily setting $\pi(i) \gets -\pi(i)$, and after having rearranged
the array restoring the original $\pi$. Even though this reduces the extra space to just one
bit per element, this might be still too much, and $\pi(i)$ might be not stored explicitly but computed on-the-fly,
as in the example of transposing a rectangular array. This motivates
the challenge of designing an efficient algorithm that only assumes access to $\pi$ through an oracle
and uses a small number of bits of extra storage. 

Designing a fast algorithm that permutes $A[1..n]$ given oracle access
to $\pi$ and using only a little extra storage is a known interview
puzzle~\cite[Sec. 6.9]{Aziz2015}. The expected solution is a quadratic-time algorithm
that identifies the cycles of $\pi$ by iterating over $i=1,2,\ldots,n$ and checking if $i$
is the smallest on its cycle in $\pi$. Having identified
such $i$, we permute the values $A[i],A[\pi(i)],A[\pi^{2}(i)],\ldots$. This uses only 
a constant number of auxiliary variables, or $\Oh(\log n)$ bits of additional space, which is
asymptotically optimal as we need to be able to specify an index consisting of $\lceil \log n \rceil$ bits.
However, its worst-case running time is quadratic. Designing a faster solution is nontrivial,
but Fich, Munro and Poblete obtained an elegant $\Oh(n\log n)$-time algorithm using
only $\Oh(\log^{2}n)$ bits of extra space~\cite{FichMP95}. 
Their approach is also based on identifying the cycles of $\pi$. This is implemented by
scanning the elements $i=1,2,\ldots,n$ while testing if the current $i$ is the leader on its cycle,
designating exactly one element on every cycle to be its leader. We call this a cycle
leaders procedure. Due to the unidirectional
nature of the input, the test must be implemented by considering the elements $i,\pi(i),\pi^{2}(i),\ldots$
until we can conclude if $i$ is the leader of its cycle. The main contribution of Fich et al.
is an appropriate definition of a leader that allows to implement such a test
while storing only $\Oh(\log n)$ indices and making the total number of accesses to $\pi$ only
$\Oh(n\log n)$.
They also show an algorithm running in $\Oh(n^2/b)$ time and $b+\Oh(\log n)$ bits of space for arbitrary $b\leq n$.

A procedure that transforms the input using only a small number of bits of extra storage is 
usually referred to as an in-place algorithm. The allowed extra space depends
on the problem, but in the most strict form, this is $\Oh(\log n)$ bits where $n$ is the size
of the input. We call such a procedure strictly in-place.
This is related to the well-studied complexity class $L$ capturing
decision problems solvable by a deterministic Turing machine with $\Oh(\log n)$ bits of
additional writable space with read-only access to the input.
There is a large body of work concerned with time-space tradeoffs assuming
read-only random access to the input. Example problems include: sorting and selection~\cite{BorodinFKLT81,BorodinC82,PagterR98,Beame91,MunroP80,Frederickson87,MunroR96,RamanR99,Chan10,ChanMR15},
constructing the convex hull~\cite{DarwishE14}, multiple pattern matching~\cite{FischerGGK15} or constructing
the sparse suffix array~\cite{GawrychowskiK17,BirenzwigeGP20}.

This raises the question: is there a deterministic subquadratic strictly in-place algorithm for permuting an array?
We provide an affirmative answer to this question by designing, for every $\eps > 0$,
a strictly in-place algorithm for this problem that works in $\Oh_{\eps}(n^{1+\eps})$ time.~\footnote{We write $\Oh_{\eps}(f(n))$ to emphasise that the hidden constant depends on $\eps$.}

\subparagraph*{Previous and related work.} Fich et al.~\cite{FichMP95} designed a cycle leaders algorithm
 that works in $\Oh(n\log n)$ time and uses $\Oh(\log^{2}n)$ bits
of extra space. Given oracle access to both $\pi$ and $\pi^{-1}$, they also show
a simpler algorithm that needs only $\Oh(\log n)$ bits of extra space and the same time. Both
algorithms have interesting connections to leader election in, respectively, unidirectional~\cite{DolevKR82,Peterson82}
and bidirectional~\cite{HirschbergS80,Franklin82} rings. While in certain scenarios, such as $\pi$ being
specified with an explicit formula, one can assume access to both $\pi$ and $\pi^{-1}$,
in the general case this is known to significantly increase the necessary space~\cite{Golynski09}.
For a random input, traversing the cycle from $i$ until we encounter a smaller element
takes in total $\Oh(n\log n)$ average time and uses only $\Oh(\log n)$ bits of space~\cite{Knuth71}.
Using hashing, one can also design an algorithm using expected $\Oh(n\log n)$ time
and $\Oh(\log n)$ bits of space without any assumption on $\pi$ (interestingly, a different application
of randomisation is known to help in leader election in anonymous unidirectional rings~\cite{ItaiR90}).
Better cycle leaders algorithms are known for some specific permutations, such as the perfect shuffle~\cite{EllisKF00}.

An interesting related question is that of inverting a given permutation $\pi$ on $n$ elements.
For example, it allows us to work with both the suffix array and the inverse suffix array without 
explicitly storing both of them~\cite{FischerIKS18}. The goal of this problem is to replace
$\pi$ with its inverse $\pi^{-1}$ efficiently while using a small amount of extra space.
El-Zein, Munro and Robertson~\cite{El-ZeinMR16}
solve this in $\Oh(n\log n)$ time using only $\Oh(\log^{2}n)$ bits of extra
space. The high-level idea of their procedure is to identify cycle leaders and invert every cycle of $\pi$ at its
leader $i$. The difficulty in such an approach is that the leader $i'$ of the inverted cycle
might be encountered again, forcing the cycle to be restored to its original state. El-Zein et al. deal with this hurdle
by temporarily lifting the restriction that $\pi$ is a permutation.
Very recently, Guśpiel~\cite{Guspiel} designed a strictly in-place algorithm for this problem
that works in $\Oh(n^{1.5})$ time. We stress that his approach does not provide a subquadratic strictly in-place
solution for identifying cycle leaders, and so does not imply such an algorithm for permuting an array.
See Table~\ref{table-with-comparison} with the summary of previous work.

\begin{table}[h]
\centering
\begin{tabular}{ |l|c|c| }
\hline
  & Permuting & Inverting \\
\hline
\multirow{2}{*}{Trivial} & \multicolumn{2}{c|}{$(n,n)$} \\
\cline{2-3}
& \multicolumn{2}{c|}{$(n^2,\log n)$} \\
\hline
    With hashing\tablefootnote{This approach runs in expected $\Oh(n\log n)$ time.} & $(n\log n,\log n)$ & - \\
\hline
\multirow{2}{*}{Fich et al. \cite{FichMP95}} & $(n^2/b,b+\log n)$ & - \\
\cline{2-3}
& $(n\log n,\log^2 n)$ & - \\
\hline
El-Zein et al. \cite{El-ZeinMR16} & - & $(n\log n,\log^2 n)$\\
\hline
Guśpiel \cite{Guspiel} & - & $(n^{1.5},\log n)$ \\
\hline
This work\tablefootnote{The constant in time and space complexity depends on $\eps$.} & \multicolumn{2}{c|}{$(n^{1+\eps},\log n)$}\\
\hline
\end{tabular}
\caption{Comparison
of different algorithms for in-place permuting and inverting a permutation.
$(f(n),g(n))$ denotes that the algorithm runs in $\Oh(f(n))$ time and uses $\Oh(g(n))$ bits of
space.
}
\label{table-with-comparison}
\end{table}

\subparagraph*{Our contribution.} Building on the approach of Fich et al. we design, for every $\eps>0$,
a cycle leaders algorithm that works in $\Oh_{\eps}(n^{1+\eps})$ time and uses $\Oh_{\eps}(\log n)$
bits of extra space. This implies a strictly in-place algorithm for permuting an array in $\Oh_{\eps}(n^{1+\eps})$
time. In other words, we show that by increasing the number of auxiliary variables to a larger constant we
can make the exponent in the running time arbitrarily close to 1. Then, we apply our improved
cycle leaders algorithm to obtain a solution for inverting a given permutation in the same time and space.
This significantly improves on the recent result of Guśpiel~\cite{Guspiel}.

\subparagraph*{Techniques and roadmap.} The main high-level idea in Fich et al. is to work with local minima,
defined as the elements $i\in E_{1}=[n]$ of $\pi_{1}=\pi$ such that $\pi^{-1}(i) > i < \pi(i)$. This is applied iteratively by defining
a new permutation $\pi_{2}$ on the set $E_{2}$ of local minima, and repeating the same construction on $\pi_{2}$.
After at most $t\leq \lfloor \log n \rfloor$ iterations, there is only the smallest element $m$ remaining and the leader
is chosen as the unique $i$ on the cycle such that $\pi_{t} \circ \ldots \circ \pi_{1}(i) = m$.
Checking if $i$ is a leader is done by introducing the so-called \emph{elbows}. For completeness, we provide
a full description of the algorithm in Appendix~\ref{se:fich}. Compared to the original version, we introduce
new notation and change some implementation details to make the subsequent modifications easier to state.
The crucial point is that the extra space used by the algorithm is bounded by a constant number of words per
iteration in the above definition. The natural approach for decreasing the space to $\Oh(\log n)$ is
to modify the definition of the local minimum to decrease the number of iterations to a constant. To this end, we work
with $b$-local minima defined as elements less than all of their $b$ successors and predecessors,
for $b=\lceil n^{\eps} \rceil$ where $\eps$ is a sufficiently small constant. This decreases the number of iterations to
$\log n / \log b = 1/\eps = \Oh_{\eps}(1)$.
There is, however, a nontrivial technical difficulty when trying to work with this idea.
In the original version, one can check $\pi_{r}(x)\in E_{r}$ should be ``promoted'' to $E_{r+1}$
by explicitly maintaining $x,\pi_{r}(x)$ and $\pi_{r}(\pi_{r}(x))$ and simply comparing
$\pi_{r}(x)$ with $x$ and $\pi_{r}(\pi_{r}(x))$. For larger values of $b$, this translates into
explicitly maintaining $x,\pi_{r}(x),\pi_{r}(\pi_{r}(x)),\ldots,\pi_{r}^{2b}(x)$ to check if $\pi_{r}^{b}(x)\in E_{r+1}$,
which of course takes too much space.
We overcome this difficulty by designing and analysing a recursive pointer.
This gives us our cycle leaders algorithm described in Section~\ref{se:fich-b}. 

Similarly as El-Zein et al.~\cite{El-ZeinMR16}, we use our cycle leaders algorithm to design
a solution to the permutation inversion problem. The high-level idea is to identify the leader $i$ of a cycle,
and then invert the cycle by traversing it from $i$. We need to somehow guarantee that the cycle is not inverted
again, but do not have enough extra space to store $j$. El-Zein et al. mark the already inverted cycles that
otherwise would be again inverted in the future by converting them to paths, that is, changing some $\pi(x)$
to undefined. This is then gradually repaired to a cycle, which requires a nontrivial interleave of four different scans.
Our starting point is a simplification of their algorithm described in Appendix~\ref{se:inverting-b1} based on encoding
slightly more (but still very little) information about the cycle. In Section~\ref{se:inverting-b} we further extend this
method to work with $b$-local minima.

\section{Preliminaries}

$[n]$ denotes the set of integers $\{1,2,\ldots,n\}$ and for a function $f$ and a nonnegative integer $k$ we
define $f^{k}(x)$ to be $x$ when $k=0$ and $f(f^{k-1}(x))$ otherwise. If $f$ is a permutation then
$f^{-1}$ denotes its inverse and we define $f^{-k}(x)$ to be $f^{-k}(x) = f^{-1}(f^{-k+1}(x))$ for $k > 0$.
Throughout the paper $\log$ denotes $\log_{2}$.

In the cycle leaders problem, we assume that the permutation $\pi$ on $[n]$ is given through an oracle that returns any
$\pi(i)$ in constant time. The goal is to identify exactly one element on each cycle of $\pi$ as a leader. All of our
algorithms follow the same left-to-right scheme: we consider the elements $i=1,2,\ldots,n$ in this
order and test if the current $i$ is the leader of its cycle by considering the elements $i,\pi(i),\pi^{2}(i),\ldots$
until we can determine if $i$ is a leader. 

\begin{algorithm}[t]
\begin{algorithmic}[1]
  \For{$i=1..n$} $\Process(i)$ \EndFor
\end{algorithmic}
\caption{A general framework of \textit{left-to-right} algorithms.}
\label{alg:left_to_right}
\end{algorithm}

Whenever we refer to a~range $x \ldots y$, we mean
$x, \pi(x), \pi(\pi(x)), \dots, y$. We will also consider ranges
$x \ldots y$ longer than a full cycle, but in such cases
there will be always an middle point $z$ (clear from the context) such that
$x \ldots y$ consists of two ranges $x \ldots z$ and $z \ldots y$,
each shorter than a full cycle. For example, in the range $\pi^{-k}(x) \ldots \pi^k(x)$ there are
elements from $\pi^{-k}(x) \ldots x$ and from $x \ldots \pi^k(x)$.
We use $\len{i}{i'}$ to denote $\min\{k > 0 : \pi^k(i) = i'\}$.
$\Min(a, b)$ naively finds the minimum between $a$ and $b$ on the same
cycle of $\pi$, that is, $\min\{a,\pi(a),\pi^{2}(a),\ldots,b\}$. If $a = b$, we
assume it computes the minimum of the full cycle of $a$.
We also create ternary $\Min$ as follows:
$\Min(x, y, z) = \min(\Min(x, y), \Min(y, z))$.
Using this notation,
the naive cycle leaders algorithm is presented in Algorithm~\ref{alg:process}.

\begin{algorithm}[t]
\begin{algorithmic}[1]
  \Function{$\Process$}{$i$}
    \If{$i=\Min(i,i)$} Report that $i$ is a leader. \EndIf
 \EndFunction
\end{algorithmic}
\caption{A naive cycle leaders algorithm.}
\label{alg:process}
\end{algorithm}

When inverting $\pi$, we assume constant-time random access to the input. Due to the final
goal being replacing every $\pi(i)$ with $\pi^{-1}(i)$, we allow temporarily overwriting $\pi(i)$ with any value from $[n]$
as long as after the algorithm terminates the input is overwritten as required.
Additional space used by our algorithms consists of a number of auxiliary variables called words, each capable
of storing a single integer from $[n]$. We assume that basic operations on such variables take constant time. We
assume that the value of $n$ is known to the algorithm.

\section{Leader election in smaller space}\label{se:fich-b}

In this section, we extend the algorithm of Fich et al.~\cite{FichMP95} to obtain, for any $\eps>0$,
a solution to the cycle leaders problem in $\Oh_\eps(n^{1+\eps})$ time and
using $\Oh_\eps(\log n)$ bits of extra space.

Let $E_1 = [n]$ and $\pi_1 = \pi$. We denote by $b=\lceil n^{\eps} \rceil$ the size of the neighborhood considered while
determining local minima and declare an element of a permutation to be a $b$-local minimum if it is
strictly smaller than all of its $b$ successors and predecessors. Then,
$E_r$ is the set of all $b$-local minima in $E_{r-1}$, that is,
$E_r=\{i\in E_{r-1} : i < \pi_{r-1}^k(i) \textrm{ for all } k \in \{-b, \dots, b\} \setminus \{0\}\}$.
We say that an element is on level $r$ if it belongs to $E_r$.
We define $\pi_r : E_r \rightarrow E_r$ as follows: $\pi_r(e)$ is the first
element of $E_r$ appearing after $e$ on its cycle in $\pi$.
For $b=1$ this is exactly the definition used by Fich et al.~\cite{FichMP95}, and for the reader who is
unfamiliar with their work we provide a recap of the algorithm and some intuition in Appendix~\ref{se:fich}.

We use the same framework in the next section for the permutation inversion problem,
and so we already introduce some additional definitions that are only used later.

\begin{definition}
A \emph{path} of $\pi$ is a partial function obtained from a cycle $C$ of $\pi$
by replacing $\pi(x)$ with $\perp$, for some $x \in C$.
\end{definition}

$\perp$ should be understood as an undefined element. For a path, $\pi_{r}$
is undefined for the last element from $E_{r}$, and similarly $\pi^{-1}_{r}$ is undefined
for the first element from $E_{r}$. When deciding if an element is a $b$-local minimum on a~path,
we disregard comparisons with such undefined elements.
Furthermore, we assume that $\pi_r(\perp) = \perp$ and $\pi^{-1}_{r}(\perp)=\perp$ for every $r$.

For each level $r$, only at most $\frac{1}{b+1}$ of its elements can
belong to the level $r+1$. For a~cycle or a~path $C$, let $t$ be the largest
number such that $|E_t \cap C| > b$. We set $t = 0$ for $|C| \le b$ and
observe that $t < \frac{1}{\epsilon}$ for every $C$.
We note that if $C$ is a~cycle then $E_{t+2} \cap C = \emptyset$, but for
$C$ being a~path $|E_{t+2} \cap C| = 1$.
Because all algorithms presented in this paper follow the framework given in
Algorithm~\ref{alg:left_to_right} and during the execution of $\Process(i)$, we
only consider the elements that can be reached from $i$, we restrict
our considerations to just one cycle or path and we are going to omit the
``$\cap C$'' part everywhere later.

\begin{definition}
 A \emph{$b$-staircase} of size $r$ from $i$ is a sequence of elements
 $(i=i_1,i_2,\ldots,i_{r+1}=m=j_{r+1},j_r,j_{r-1},\ldots,j_1=i')$ such that
 $i_k,j_k\in E_k$, for $k\in[r+1]$ and $i_{k+1}=\pi_k^b(i_k)$, $j_k=\pi_k^b(j_{k+1})$ for $k\in[r]$.
 Elements $i$, $m$ and $i'$ are~called \emph{the start}, \emph{the middle}
 and \emph{the end} of the $b$-staircase, respectively.
 The part of the $b$-staircase from the start to the middle is called
 its \emph{left part} and the part from the middle to the end is called its
 \emph{right part}.
\end{definition}

\begin{definition}
  An \emph{almost $b$-staircase} of size $r$ from $i$ is a sequence of elements
  $(i=i_1,i_2,\ldots,i_{r+1}=m=j_{r+1},j_r,j_{r-1},\ldots,j_1=i')$, such that
  $i_k,j_k\in E_k$, for $k\in[r]$ and $i_{k+1}=\pi_k^b(i_k)$ and
  $j_k=\pi_k^b(j_{k+1})$ for $k\in[r]$.
\end{definition}

\begin{definition}
  An almost $b$-staircase of size $r$
  $(i=i_1,i_2,\ldots,i_{r+1}=m=j_{r+1},j_r,j_{r-1},\ldots,j_1=i')$
  is called \emph{proper} if $|E_r| > b$.
\end{definition}

\begin{definition}\label{def:BestBStaircase}
  A~$b$-staircase of size $r$ from $i$ is called the \emph{best $b$-staircase
  from $i$} if there is no proper almost $b$-staircase of size $r+1$ from $i$
  (possibly there is no best $b$-staircase from $i$).
\end{definition}

\begin{definition}
 An element $i$ is the leader of its cycle if the best $b$-staircase from $i$
 exists and its middle $m$ is the minimum on the cycle.
 \label{def:leader-b}
\end{definition}

\begin{lemma}
  There is exactly one leader on any cycle of a permutation.
\end{lemma}
\begin{proof}
  Let $m$ be the minimum on the cycle and $t$ be the largest size of a proper
  $b$-staircase on the cycle.
  Consider the $b$-staircase $B$ of size $t$ with its start in
  $i = \pi_1^{-b}(\ldots(\pi_{t-1}^{-b}(\pi_t^{-b}(m)))\ldots)$.
  As $m$ is the minimum on the cycle, there is no proper staircase of size $t+1$
  from $i$, so $B$ is the best $b$-staircase and hence $i$ satisfies
  Definition \ref{def:BestBStaircase} of the leader.
  As all the best $b$-staircases are of the same size we cannot have more than one leader on
  the cycle.
\end{proof}

It may be that a~$b$-staircase is longer than a full cycle if it
visits some elements twice (before and after reaching the middle of the
$b$-staircase). All elements occur at most once in each part of the
$b$-staircase.
We allow computing $\Min(x, y)$ for $x = \perp$ or $y = \perp$.
In such a case, the minimum is computed from the beginning or to the end of the
path, respectively.

\begin{lemma}
  Consider $m\in E_r$ on a cycle or a path where $|E_r|>b$.
  Then $m$ is a~$b$-local minimum on level $r$ in $\pi$ if and only if
  $m = \Min(\pi_r^{-b}(m), m, \pi_r^b(m))$.
  \label{le:b-local-min}
\end{lemma}
\begin{proof}
  ($\Leftarrow$) Trivial. ($\Rightarrow$) Proof by induction on level $r$.
  The base case $r = 1$ is immediate.

  For $r > 1$, assume that $m$ is~a $b$-local minimum on level $r$.
  We consider the elements in $\pi_{r}^{-b}(m)\ldots \pi_{r}^{b}(m)$, appropriately
  truncated if we are on a path. Among these elements, only $m_k = \pi_r^k(m) \neq \perp$
  for $k \in \{-b, \dots, b\}$ are $b$-local minima on level $r-1$.
  All of them are also larger than $m$ (if $k \neq 0$). By the induction
  hypothesis, for all $k \in \{-b, \dots, b\}$ (such that $m_k \neq \perp$)
  any element between $\pi_{r-1}^{-b}(m_k)$ and
  $\pi_{r-1}^b(m_k)$ on the cycle/path is larger than or equal to $m_k$ and
  hence also larger than $m$.
  We call the above ranges the \emph{ranges around $m_k$} and they are
  represented as segments in Figure~\ref{fig:proof-min-range}.

  \FIGURE{h}{0.98}{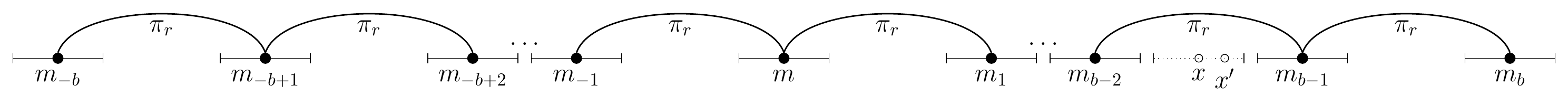}{
    Filled nodes represent elements $m_k=\pi_r^k(m)$ for $k\in\{-b,\ldots,b\}$
    and each segment represents the range around $m_k$.
    All the elements in the ranges (different than $m$) are larger than $m$ by
    induction hypothesis.
  }

  Consider an element $x$ in $\pi_r^{-b}(m) \ldots \pi_r^b(m)$
  which is not in the range around $m_k$ for any $k$.
  In Figure~\ref{fig:proof-min-range}, this corresponds to $x$ in a gap between
  the ranges.
  Clearly, $x \not \in E_r$, so there is an $\ell <r$ such that
  $x\in E_\ell \setminus E_{\ell+1}$.
  Hence $x$ is larger than some $x'\in E_{\ell'}\setminus E_{\ell'+1}$ between
  $\pi_\ell^{-b}(x)$ and $\pi_\ell^b(x)$ for $\ell < \ell' \le r$.
  If there is $k$ such that $x'$ is in the range around $m_k$, then $x'>m_k>m$ and
  the lemma follows.
  Otherwise, observe that $x'$ belongs to the same gap as $x$ (see Figure
  \ref{fig:recursive-elbows} for more detailed explanation why) so we can apply
  the same reasoning for $x'$ instead of $x$.
  Because the considered values are always decreasing and each gap contains
  only finite number of elements, finally we obtain a value in one of the ranges
  around $m_k$ and conclude that $x$ is larger than one of the $m_k$ and also
  than $m$.
  Hence, $m$ is the smallest element in $\pi_r^{-b}(m) \ldots \pi_r^b(m)$.
\end{proof}
The lemma enables us to check if an element $m$ is a $b$-local minimum on
level $r$ without knowing which of the elements in
$\pi_r^{-b}(m)\ldots\pi_r^b(m)$ belong to the level $r$.

Fich et al. develop the procedure $\Next(r)$ which computes, for an element on $E_r$,
its successor on $\pi_r$.
We start with extending their idea to~work with larger values of the parameter $b$.
For $b = \lceil n^{\epsilon} \rceil$, the number of elements to compare with is
no longer constant.
We~define a recursive pointer $\Ptr$ with the following fields:
\begin{itemize}
  \item $r$ -- the level of the pointer,
  \item $e$ -- an element of $E_r$ pointed at by the pointer,
  \item $x$ -- $\Ptr$ of level $r-1$ pointing to $e$ (or $\NULL$ if $r = 1$),
  \item $y$ and $z$ -- $\Ptr$s of level $r-1$ both pointing to $\pi_{r-1}^b(e)$
  (or $\NULL$ if $r = 1$).
\end{itemize}

$\Ptr$ has the $\Advance$ method which moves the pointer
from $e$ to $\pi_r(e)$ and updates $x$, $y$ and $z$ accordingly.
See Figure \ref{fig:ptr}.
Before and after calling $\Advance$, $y$ is equal to $z$, but during
the execution they are different.
In fact, it is enough to only store $e$, $x$ and $z$ in every $\Ptr$
and keep a local $y$ reused between every call to $\Advance$ method.

\FIGURE{h}{1}{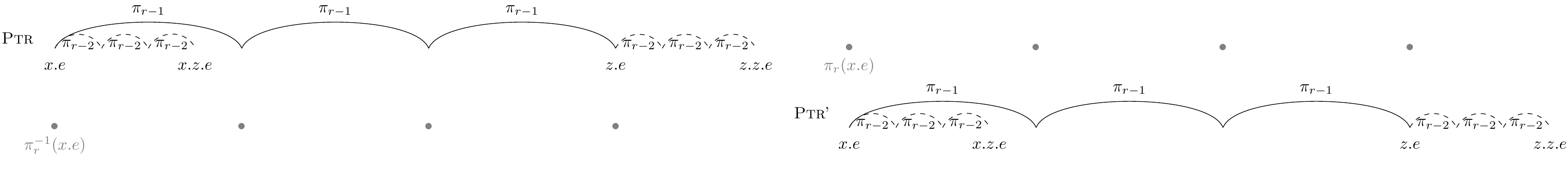}{
  A sketch of $\Ptr$ structure for $b = 3$. On the bottom there is $\Ptr'$, the result of executing the $\Advance$ method on $\Ptr$.
  The gray dots are not part of pointer, they only show the alignment between $\Ptr$ and $\Ptr'$.
}

\begin{algorithm}[t]
\begin{algorithmic}[1]
  \Function{$\Ptr::\Advance$}{}
    \If{$\this.r = 1$}
      \If{$\pi(\this.e) = \perp$} \textbf{abort}
      \label{li:abort2}
      \EndIf
      \State $\this.e \gets \pi(\this.e)$
      \State \Return
    \EndIf
    \For{$j=1..b$}
      $\this.z.\Advance()$
    \EndFor
    \label{li:z-advanced}
    \While{$\this.y \neq \Min(\this.x.e, \this.y.e, \this.z.e)$}\label{li:check-min}
      \State $\this.x.\Advance()$
      \State $\this.y.\Advance()$
      \State $\this.z.\Advance()$
    \EndWhile
    \State $\this.x \gets \this.y$
    \State $\this.e \gets \this.x.e$
    \State $\this.y \gets \this.z$
  \EndFunction
\end{algorithmic}
\caption{Implementation of the $\Advance$ method in the $\Ptr$ structure. It
moves the $\Ptr$ object onto the next element on the same level.}
\label{alg:ptr_advance}
\end{algorithm}

We now analyse $\Ptr::\Advance$ implemented in
Algorithm~\ref{alg:ptr_advance} for a pointer of level $r$.
At the end of the function, $e$ is updated to point to the element
$x.e \in E_r$.
As we have to also deal with paths, we add the check for $\perp$ in line
\ref{li:abort2}.
For $r \ge 2$, $z$ is moved forward $b$ times along $\pi_{r-1}$, so after line
\ref{li:z-advanced}, $y$ and $z$ point to $\pi_{r-1}^b(x.e)$ and
$\pi_{r-1}^{2b}(x.e)$, respectively.
In line~\ref{li:check-min} we compare $y.e$ with all elements
between $\pi_{r-1}^{-b}(y.e)$ and $\pi_{r-1}^b(y.e)$ and all three pointers are
simultaneously advanced to the next elements along $\pi_{r-1}$ until $y$ points
to a $b$-local minimum on level $r$.
By Lemma~\ref{le:b-local-min} this is equivalent to checking if $y.e$ is a
$b$-local minimum on level $r-1$ as long as $|E_{r-1}| > b$.

\begin{algorithm}[t]
\begin{algorithmic}[1]
  \Function{$\BestBStaircase$}{$i$}
    \State $p \gets \Ptr(e \gets i,~ r \gets 1,~ x \gets \NULL,~ y \gets \NULL,~ z \gets \NULL)$
    \State $x \gets i$
    \For{$j=1..b$}
      \State $x \gets \pi(x)$
      \If{$x \in \{\perp, i\}$}
        \Return $p$
      \EndIf
    \EndFor
    \While{$\true$}
      \State $g \gets p$
      \State $m_e \gets p.e$
      \For{$j=1..2b$}
        \State \Try{
          $p.\Advance()$
        }
        \Catch{ \textit{/abort/} }{
          \Return $g$
        }
        \If{$j \le b\textbf{ and }p.e = g.e$}
          \Return $g$
        \EndIf
        \If{$j = b$}
          $m_p \gets p$
        \EndIf
        \If{$p.e < m_e$}
          $m_e \gets p.e$
        \EndIf
      \EndFor
      \If{$m_p.e = m_e$}
        $p \gets \Ptr(e \gets m_p.e,~ r \gets p.r + 1,~ x \gets m_p,~ y \gets p,~ z \gets p)$
      \Else
        ~\Return \NULL \label{li:abort3}
      \EndIf
    \EndWhile
  \EndFunction
\end{algorithmic}
\caption{Constructs the best $b$-staircase from $i$ (if any).}
\label{alg:best_bstaircase}
\end{algorithm}

In Algorithm~\ref{alg:best_bstaircase} we implement $\BestBStaircase(i)$ that constructs the best
$b$-staircase from $i$ returning $\Ptr$ structure pointing to the middle of the
staircase.
Before running the main loop, the algorithm checks if $|E_1| > b$, to make sure it is allowed to execute $\Advance$
method on $p$. We start with a $b$-staircase of size $0$.
Each iteration of the main loop starts with $p$ of level $r$ representing the
right part of the~$b$-staircase of size $r-1$ from $i$.
The invariant $|E_r| > b$ is preserved in the main loop.
During each iteration, $p$ is subject to change and the updated $p$
may not represent a (proper) $b$-staircase, so at the beginning of the
iteration, it is stored in $g$.
Then, $\Advance$ method is called $2b$ times on $p$ and enumerates the elements
from the set $S = \{\pi_r^k(e)$ : $k \in \{1, \dots, 2b\}\}$. The pointer $m_p$
is stored after the $b$-th $\Advance$ so it points to $\pi_r^b(e)$, that is the
middle of the almost staircase of size $r$.
If $m_p.e$ is not the smallest in $S$, then the best $b$-staircase from $i$ does
not exist and we terminate in line~\ref{li:abort3}.
However, if any of $\pi_r^k(e)$ for $k \in \{1, \dots, b\}$ is equal to $e$,
then $|S| \le b$, so there is no proper almost staircase of size $r$ from $i$
(both left and right part would self-overlap) and we return $g$ as the best
staircase.
The same happens for paths in case any $\Advance$ call aborts.
Only when $|E_r| > b$ and the new middle $m_e$ is less than all $b$
pairwise-distinct left and $b$ pairwise-distinct right neighbours on
$E_r$, the algorithm extends the $b$-staircase to size $r$ and proceeds to the
next iteration with $p$ of level $r+1$.

We are ready to provide an algorithm for the cycle leaders problem. We proceed
as in Algorithm~\ref{alg:left_to_right}, but we alter the
$\Process$ function to work with the new definition of leader, see
Algorithm~\ref{alg:process_with_b_best}.
\begin{algorithm}[t]
\begin{algorithmic}[1]
  \Function{$\Process$}{$i$}
    \State $p \gets \BestBStaircase(i)$
    \If{$p \ne \NULL\textbf{ and }\Min(i, i) = p.e$}
      Report that $i$ is the leader
    \EndIf
 \EndFunction
\end{algorithmic}
\caption{Modified function for checking if $i$ is the leader of its cycle.}
\label{alg:process_with_b_best}
\end{algorithm}
$\BestBStaircase(i)$ returns $\Ptr$ representing the right part of the
best $b$-staircase from~$i$. Its middle can be read from the field
$e$ of the returned $\Ptr$.

\begin{theorem}
  For every $\eps>0$, there exists an algorithm for reporting leaders of
  a~permutation $\pi$ on $n$ elements in $\Oh_\eps(n^{1+\epsilon})$
  time using $\Oh_\eps(\log n)$ additional bits of space.
  \label{thm:alg_rep_leaders}
\end{theorem}
\begin{proof}
We analyse Algorithm~\ref{alg:left_to_right} with implementation of
$\Process$ provided in Algorithm~\ref{alg:process_with_b_best}.
Each $\Ptr$ of level $r$ has three $\Ptr$s on level $r-1$, so there are
$\sum_{k=0}^{t} 3^k = \Oh(3^t)$ pointers in total.
As $t \le \frac{1}{\epsilon}$, our algorithm uses
$\Oh(3^{\frac{1}{\epsilon}} \cdot \log n)=\Oh_\eps(\log n)$ extra bits of space.

The total time for reporting leaders is dominated by the time spent in
$\Ptr::\Advance$.
We sum up (for all created pointers) just the time inside $\Advance$ call of
each $\Ptr$ without recursive calls. The time for recursive calls will be
accounted to the descendant pointers.
Inside $\Advance$, the time is dominated by $\Min$ executions.
As each advanced position is covered by at most $2b+1$ $\Min$s, if a pointer $p$
proceeds $k$ steps along $\pi$ (starting with $p.x.e = s$ and ending
with $p.z.e = \pi^k(s)$) during $\Advance$,
the total time spent on executions of $\Min$ for the call is $\Oh(kb)$.

Now we count the total length of traversals for all created pointers.
This, multiplied by $\Oh(b)$, is the final complexity of the algorithm.
Recall that for every element, we check if it is possible to create the best
$b$-staircase from it by constructing proper almost $b$-staircases and
validating if these are $b$-staircases. Thus, if it is possible to
construct a $b$-staircase of size $r$ and not $r+1$ from $i$, 
every pointer traverses (in the worst possible case) all elements of the
almost $b$-staircase of size $r+1$ from $i$.
We charge the traversed range to the pair $(r, m)$,
where $m \in E_{r+1} \setminus E_{r+2}$ is the middle of the $b$-staircase of size $r$ from $i$.
Observe that $i$ is between $\pi_{r+1}^{-1}(m)$ and $m$ and the execution of
$\BestBStaircase(i)$ finishes before reaching $\pi_{r+1}^{2b+1}(m)$, see Figure
\ref{fig:range-of-middle-ptr}. Thus, for each $r$, every element
belongs to $\Oh(b)$ ranges charged to the elements from $E_{r+1} \setminus E_{r+2}$.
The total length of all such ranges is therefore $\Oh(nb)$.
Recall that there are $\Oh(3^r) = \Oh(3^t) = \Oh_{\epsilon}(1)$ pointers on each
level $r$, so summing over all $r \le t$ we obtain that the running time of
Algorithm~\ref{alg:process_with_b_best} is
$\Oh(\frac{1}{\epsilon} \cdot 3^t \cdot nb^2)=\Oh_{\epsilon}(n^{1+2\epsilon})$.
By adjusting $\epsilon$, the running time becomes $\Oh_{\epsilon}(n^{1+\epsilon})$.
\end{proof}

\FIGURE{h}{1}{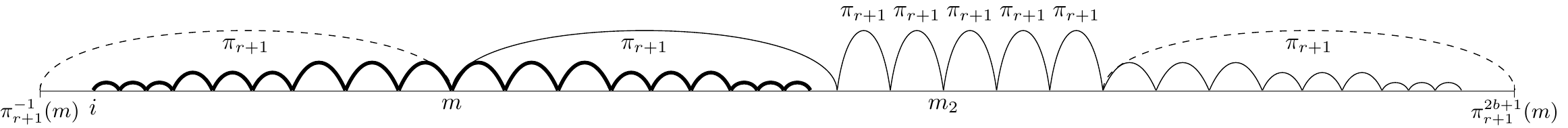}{
  The maximum potential traversal range for a~middle $m$ of a~$b$-staircase
  ($b = 3$) of size~$r$. The~$b$-staircase of size $r$ with the middle element
  $m$ is drawn with a~bold line and the~almost $b$-staircase of size $r+1$ with
  new middle $m_2$ is drawn with a~solid line.
}

\section{Inverting permutation in smaller space}\label{se:inverting-b}

In this section, we extend the results of Section \ref{se:fich-b}
to obtain an $\Oh(n^{1+\epsilon})$-time algorithm for
inverting permutations using $\Oh_\eps(\log n)$ bits. As this builds
on the algorithm given by El-Zein et al.~\cite{El-ZeinMR16}, the reader
unfamiliar with their work is encouraged to first read Appendix
\ref{se:inverting-b1}
where we describe its simpler version.

The natural idea is to invert each cycle upon reaching its leader.
Unfortunately, after inverting a cycle for its leader $i$
we might obtain a cycle with leader $i'$ which might be greater than $i$ and
we do not want to invert this cycle again.

After inverting a~cycle, the middle of the best $b$-staircase remains
the same and the only best $b$-staircase on the inverted cycle is
exactly the reverse of the previous best $b$-staircase. We can compute
the end $i'$ of the largest $b$-staircase constructed from $i$ by implementing
$\GetEnd$ method of $\Ptr$ returning the end of the staircase as follows:
$p.\GetEnd() = p.z.\GetEnd()$ if $p.z\ne\NULL$ and $p.\GetEnd()=p.e$ otherwise.

\begin{definition}
  A~cycle $C$ of $\pi$ is~\emph{hard} if~for the best $b$-staircase
  $(i, i_2, i_3, \dots, i_r, m, j_r, \dots, j_3, j_2, i')$ we have $i' > i$.
  Otherwise, $C$ is \emph{easy}.
\end{definition}

For $x$ on a~cycle and $y = \pi(x)$, by setting $\pi(x) \gets \perp$
we obtain a~path $P$ with $y$ as its first element and $x$ as the last
element. We call this operation \emph{cutting before $y$}.
For hard cycles, we are going to cut the inverted cycle before its new leader
and store null at the end of the path.
As in \cite{El-ZeinMR16}, we use multiple types of nulls that encode some
additional information. They can be efficiently simulated in our model using the
idea of El-Zein et al.~\cite{El-ZeinMR16}. For completeness, a proof can
be also found in Appendix~\ref{se:inverting-b1}.

\begin{lemma}[\cite{El-ZeinMR16}]
 For any $k\in[n]$, it is possible to simulate an extension of the range of stored values of
 $\pi$ from $[n]$ to $[n] \cup \{\perp_1,\ldots,\perp_k\}$ with each value from
 $[n]$ occurring at~most $c$ times in $\pi$ using $\Oh(ck\log n)$ bits with
 $\Oh(c)$-time overhead.
 \label{lemma:nulls-extension-paper}
\end{lemma}

\begin{definition}
  \emph{Sigma} is~a~function $\pi : C \rightarrow C$, such that there
  is at~least one element $y \in C$ with $\{\pi^k(y)\ :\ k \in \{0, 1, 2, \dots\}\} = C$.
  We call sigma \emph{proper} if it is not a~cycle.
\end{definition}

Depending on the number of times $k$ an element that looks like a~possible
leader was found for a~cycle, we can be in one of the following states:
\begin{itemize}
  \item not inverted cycle (for $k = 0$),
  \item inverted path (for $k = 1$ and only if the cycle was hard),
  \item inverted sigma (for easy cycles if $k = 1$ and for hard cycles if
  $k \ge 2$).
\end{itemize}

We use the name \emph{component} for both paths and proper sigmas.
In our algorithm, we always have $\pi$ consisting only cycles and components.
For a~path we call the element without the predecessor the~\emph{start}
of the~path and the element pointing to null the~\emph{end} of the~path.
Similarly, in a~proper sigma we call element $y_1$ without
predecessor the~\emph{start} of the sigma and among the two elements $x_1$,
$x_2$ such that $\pi(x_1) = \pi(x_2) = y_2$ satisfying
$\len{y_1}{x_1} < \len{y_1}{x_2}$, $x_2$ is called the~\emph{end} and $y_2$ the
\emph{intersection} of the sigma.
A part of sigma before the intersection (from $y_1$ to $x_1$) is called
the~\emph{tail} and the rest (a~cycle with $y_2$ and $x_2$) is the~\emph{loop} of
sigma.

Notice that if we point the end of a component to its start, then it
becomes a~cycle.
We call such operation \emph{fixing} a component and the component is \emph{fixed}.

As we observed earlier, $b$-staircases constructed on cycles may be
self-overlapping, but staircases constructed on paths may not. For sigmas, the
situation is not that simple:
\begin{itemize}
  \item from the perspective of an~element on the loop of a~sigma, it is not
  easy to~recognise if~it belongs to~a cycle or~to~the loop of~a~sigma, and
  so we allow self-overlapping staircases in this case (as~in any cycle),
  \item from the perspective of an~element on the tail of a~sigma, it is
  possible to access the intersection of sigma (we show how later in this
  section), and we disallow traversing any element twice by pretending that the
  connection from the end to the intersection of the sigma points to
  a~null, effectively being the end of a path.
\end{itemize}

\begin{definition}
  The \emph{leader} of a component is the element that would be the leader in
  its fixed component.
\end{definition}

\begin{definition}
  The \emph{rank} of~element $i$ on~a~path or on~a~tail of a~sigma in $\pi$
  is~the size of the best $b$-staircase from $i$. If there is no~best
  $b$-staircase from $i$ then its rank is~undefined.
\end{definition}

\begin{definition}
  An element $i$ is \emph{outstanding} if has the largest rank among the elements
  on its component.
\end{definition}

Our algorithm follows the left-to-right framework of
Algorithm~\ref{alg:left_to_right}.
It maintains the following invariants:
\begin{enumerate}
  \item the start of a component is its leader (and is outstanding),
  \item the intersection of a~sigma is the only element from which there is
  a~best staircase on the loop of the~sigma (and is already processed).
\end{enumerate}

The implementation of the procedure inverting permutations is provided in
Algorithm~\ref{alg:process_with_b_best_inverse}.
Intuitively, the algorithm acts differently depending if the element $i$ is on
a cycle, a path or the tail of a sigma.
It would be convenient to detect it at the beginning of processing the element
$i$, but this would take too much time.
However, all the cases start with computing $\BestBStaircase(i)$ and terminate
if the computation failed or aborted.
As we can only find the best $b$-staircase after processing the whole cycle or
component, we can augment this procedure to (if the best $b$-staircase from $i$
exists) also return the element $a$ which is either predecessor of $i$ (for $i$
on the cycle) or the end of the component (for $i$ on a path or the tail of a
proper sigma).
This is not obvious only for proper sigmas and we use the following lemma:
\begin{lemma}
  Given an element $i$ on a sigma, it is possible to find the intersection of
  the sigma or report that $i$ is on the loop in $\Oh(\ell+p)$ time, where
  $\ell$ and $p$ are lengths of loop and tail of the sigma, respectively.
  \label{lemma:intersection-of-sigma-paper}
\end{lemma}
This was described by El-Zein et al.~\cite{El-ZeinMR16}, and for completeness,
we provide the proof and implementation of $\FindIntersection$ in
Appendix~\ref{se:tortoise-and-hare}. We assume that $\FindIntersection$ can be
also run on paths and might abort during the evaluation of $\pi$. In this case,
it returns $a$ for which $\pi(a) = \perp$.
If we started from the tail of a~sigma, we can find its the end by using
Lemma~\ref{lemma:intersection-of-sigma-paper} and traversing the loop once more
to find the predecessor on the loop. Now we describe how the algorithm handles
all the cases mentioned earlier.

When the leader $i$ of a cycle is found for the first time, the algorithm finds
the end $i'$ of the corresponding $b$-staircase and inverts the cycle, so $i'$
becomes the leader of the inverted cycle. If $i' > i$, then the cycle is hard,
and we change it into a path by cutting before $i'$.
The rank of $i'$ on the~path is stored in the type of the null.

The~path remains in this state at least until we consider one of its outstanding
elements.
To check if $i$ on the~path is outstanding, we first check if it is
possible to~create the best $b$-staircase from $i$. When the best $b$-staircase
is found on the path, we check if the rank of $i$ matches the one
saved in the null and if $i$ is the leader on the loop of sigma created by
pointing the end of the path to $i$.
If this is the case, we preserve the change, otherwise undo.

If the construction of $b$-staircase starts from an element
$i$ on the loop of a~proper sigma, there is no easy way to distinguish it from
a~cycle.
Thus, it is not clear if we can use Algorithm~\ref{alg:best_bstaircase} to
verify if an element is outstanding. Fortunately, the rank of any outstanding
element on the~tail of sigma equals the
rank of the leader of the proper sigma.

If the algorithm finds an outstanding
element $i$, it creates a temporary, larger sigma, by pointing the end of the
sigma to $i$, then checks if there is the best $b$-staircase from $i$ on the
loop of sigma. If so, the change is in effect or is reverted otherwise. Because
of this check, we have no other elements than $i$ with the best
staircase on the loop of sigma.
See Figure \ref{fig:final-lifecycle}.

\FIGURE{t}{1}{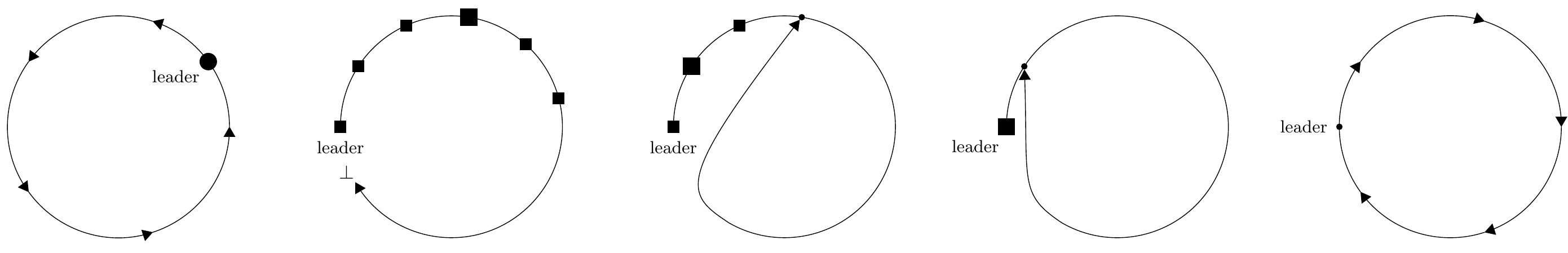}{Lifecycle of a~hard cycle during the execution
of the algorithm (chronologically from left to right). Circles are the
elements from which there is a best staircase, squares are outstanding
elements of~each component and large nodes are elements that triggered the
component update.}

Let $c$ be the intersection of the sigma. From the invariants of our algorithm
we have that $c$ was an outstanding element in the component before its last
update and its rank $t$ in the previous component was equal to the rank of the
leader.
We can compute $t$ by running $\BestBStaircase(c)$ and then we can treat the
sigma as if it was a path by assuming that the end of the sigma points to
$\perp_t$ instead of $c$.

Notice that if we were always trying to find the intersection of sigma, this
could take much more time than the execution of $\BestBStaircase(i)$, which
often terminates early.
We also need to ensure that successful call of $\BestBStaircase(i)$ never
reaches the intersection of the sigma from the end (it needs to be earlier
updated to $\perp_t$).
We execute procedures $\BestBStaircase(i)$ and $\FindIntersection(i)$
in parallel in such a way that for every step of the former procedure, we
execute a constant number of steps of the latter.
Additional time spent on finding intersections is bounded by a constant
multiple of the time of $\BestBStaircase(i)$ and can be omitted in the
complexity analysis.

\begin{algorithm}[t]
\begin{algorithmic}[1]
  \Function{$\Process$}{$i$}
    \State $(p, a) \gets \BestBStaircase(i)$ \label{li:staircase_line}
    \If{$p \ne \NULL$}
      \If{$\pi(a) = i\textbf{ and }\Min(i,p.e,a) = p.e$}
        \Comment $i$ is on a cycle.
        \State $i' \gets p.\GetEnd()$
        \State $a' \gets \pi(i')$
        \Comment $a'$ is the predecessor of the new leader in the inverted cycle.
        \State $\InvertCycle(i)$
        \If{$i' > i$}
          \State $\pi(a') \gets \perp_?$
          \State $(p', \_) \gets \BestBStaircase(i')$ \label{li:staircase_a}
          \State $r' \gets p'.r$
          \Comment $r'$ is the rank of the new leader in the path.
          \State $\pi(a') \gets \perp_{r'}$
        \EndIf
      \Else
        \If{$\pi(a) \in [n]$}
          \Comment $i$ is on a tail of sigma.
          \State $(p', \_) \gets \BestBStaircase(\pi(a))$ \label{li:staircase_b}
          \State $w \gets \perp_{p'.r}$
        \Else
          \Comment $i$ is on a path.
          \State $w \gets \pi(a)$
        \EndIf
        \If{$\perp_{p.r} = w$}
          \State $s_a \gets \pi(a)$
          \State $\pi(a) \gets i$
          \State $(p'', \_) \gets \BestBStaircase(i)$ \label{li:staircase_c}
          \If{$p'' = \NULL\textbf{ or }\Min(i,i) \neq p''.e$}
            $\pi(a) \gets s_a$
          \EndIf
        \EndIf
      \EndIf
    \EndIf
  \EndFunction
\end{algorithmic}
\caption{Final version of the function to process an element in order to
invert a permutation.}
\label{alg:process_with_b_best_inverse}
\end{algorithm}

The correctness of the algorithm follows from the two invariants listed
earlier. Each cycle is inverted exactly once when the leader of the cycle is
found.
If the cycle was easy, processing elements on the inverted cycle will not
find any leaders.
If the cycle was hard, it is replaced with a~path starting at the leader of the
inverted cycle (invariant 1 is satisfied).
On a~path, there might be several outstanding elements, any might
potentially trigger the replacement of the path into a~sigma.
The proper sigma starts with the same element as~the path from which it was
constructed and after enlarging a~proper sigma to~another proper one, its start
does not change, so invariant 1 holds.
Because created intersections are always chosen to be outstanding elements
on the tail of sigma and we later check if they allow creating the best staircase
on the loop of the enlarged sigma, invariant 2 also holds.
Sigma can be enlarged many times, but it becomes a~cycle when its start is
processed (this follows immediately from invariants 1 and 2). The resulting
cycle is the same as it was just after inverting and before cutting it into
the path, so there are no more leaders to be found.
Note that elements can become outstanding only when a~path is created
and every element on the loop of sigma is not outstanding, so the
number of outstanding elements always strictly decreases whenever a new sigma
is constructed.
To estimate the running time, we state the following property:

\begin{lemma}
  There are $\Oh(b)$ outstanding elements on every path.
  \label{lemma:outstanding-count}
\end{lemma}
\begin{proof}
  Assume there are (at~least) $2b + 2$ outstanding elements, denoted
  $i_1, i_2, \dots, i_{2b+2}$ in this order on the path, each of rank $t$.
  Let $m_k$ be the middle of $b$-staircase from $i_k$. $m_{k}$s are
  pairwise-distinct, and it can be verified that
  the elements appear in the following order on the path:
  $i_1, m_1, i_2, m_2, \dots, i_{2b+2}, m_{2b+2}$.
  Consequently, we can construct a proper almost $b$-staircase of size $t+1$
  from $i_1$, with the middle $m_{b+1}$ and ending before $m_{2b+2}$ on the
  path. This contradicts the assumption that $i_1$ is of rank $t$, as either it is of
  rank at least $t+1$ or there is no best $b$-staircase from $i_1$.
\end{proof}

\begin{theorem}
  For every $\eps>0$, there exists an algorithm for inverting a permutation
  $\pi$ on $n$ elements in $\Oh_\eps(n^{1+\epsilon})$
  time using $\Oh_\eps(\log n)$ additional bits of space.
\end{theorem}
\begin{proof}
  We estimate the running time of our algorithm.
  Any cycle is transformed $\Oh(b)$ times to intermediate components until it
  becomes a~cycle with its leader already considered by earlier $\Process$ call
  (from Lemma~\ref{lemma:outstanding-count}).

  We bound the running time in each cycle and intermediate component separately:
  the total time spent on checking for all elements of the structure if there is
  a~best $b$-staircase from each of these elements (line \ref{li:staircase_line}
  of Algorithm~\ref{alg:process_with_b_best_inverse}) is always no larger than
  the time it would take in the~fixed component. $\Min$ and $\InvertCycle$ are
  executed only for elements for which a~cycle was already traversed, so they do
  not influence the time complexity.
  An additional run of $\BestBStaircase$ (line \ref{li:staircase_a}) is executed
  only once for a cycle, but for a~start of best $b$-staircase on a~tail of
  sigma (lines \ref{li:staircase_b} and \ref{li:staircase_c}) may happen many
  times. These additional checks are only recomputing the work already
  done for the previous component, as the only elements reached are on~the loop
  of~the sigma.
  Overall, the time complexity is $\Oh_\eps(n^{1+\epsilon})$ for appropriate
  $\epsilon$.
  The amount of additional space follows from the proof of
  Theorem~\ref{thm:alg_rep_leaders}.
\end{proof}

\bibliography{biblio.bib}

\begin{thebibliography}{10}

\bibitem{Aziz2015}
Adnan Aziz, Tsung-Hsien Lee, and Amit Prakash.
\newblock {\em Elements of Programming Interviews in Java: The Insiders'
  Guide}.
\newblock CreateSpace Independent Publishing Platform, USA, 2015.

\bibitem{Beame91}
Paul Beame.
\newblock A general sequential time-space tradeoff for finding unique elements.
\newblock {\em {SIAM} J. Comput.}, 20(2):270--277, 1991.

\bibitem{BirenzwigeGP20}
Or~Birenzwige, Shay Golan, and Ely Porat.
\newblock Locally consistent parsing for text indexing in small space.
\newblock In {\em 31st {SODA}}, pages 607--626. {SIAM}, 2020.

\bibitem{BorodinC82}
Allan Borodin and Stephen~A. Cook.
\newblock A time-space tradeoff for sorting on a general sequential model of
  computation.
\newblock {\em {SIAM} J. Comput.}, 11(2):287--297, 1982.

\bibitem{BorodinFKLT81}
Allan Borodin, Michael~J. Fischer, David~G. Kirkpatrick, Nancy~A. Lynch, and
  Martin Tompa.
\newblock A time-space tradeoff for sorting on non-oblivious machines.
\newblock {\em J. Comput. Syst. Sci.}, 22(3):351--364, 1981.

\bibitem{Chan10}
Timothy~M. Chan.
\newblock Comparison-based time-space lower bounds for selection.
\newblock {\em {ACM} Trans. Algorithms}, 6(2):26:1--26:16, 2010.

\bibitem{ChanMR15}
Timothy~M. Chan, J.~Ian Munro, and Venkatesh Raman.
\newblock Finding median in read-only memory on integer input.
\newblock {\em Theor. Comput. Sci.}, 583:51--56, 2015.

\bibitem{DarwishE14}
Omar Darwish and Amr Elmasry.
\newblock Optimal time-space tradeoff for the {2D} convex-hull problem.
\newblock In {\em 22nd {ESA}}, volume 8737 of {\em Lecture Notes in Computer
  Science}, pages 284--295. Springer, 2014.

\bibitem{DolevKR82}
Danny Dolev, Maria~M. Klawe, and Michael Rodeh.
\newblock An ${O}(n\log n)$ unidirectional distributed algorithm for extrema
  finding in a circle.
\newblock {\em J. Algorithms}, 3(3):245--260, 1982.

\bibitem{El-ZeinMR16}
Hicham El{-}Zein, J.~Ian Munro, and Matthew Robertson.
\newblock Raising permutations to powers in place.
\newblock In {\em 27th {ISAAC}}, volume~64 of {\em LIPIcs}, pages 29:1--29:12.
  Schloss Dagstuhl - Leibniz-Zentrum fuer Informatik, 2016.

\bibitem{EllisKF00}
John~A. Ellis, Tobias Krahn, and Hongbing Fan.
\newblock Computing the cycles in the perfect shuffle permutation.
\newblock {\em Inf. Process. Lett.}, 75(5):217--224, 2000.

\bibitem{FichMP95}
Faith~E. Fich, J.~Ian Munro, and Patricio~V. Poblete.
\newblock Permuting in place.
\newblock {\em {SIAM} J. Comput.}, 24(2):266--278, 1995.

\bibitem{FischerGGK15}
Johannes Fischer, Travis Gagie, Pawe{\l} Gawrychowski, and Tomasz Kociumaka.
\newblock Approximating {LZ77} via small-space multiple-pattern matching.
\newblock In {\em 12th {ESA}}, volume 9294 of {\em Lecture Notes in Computer
  Science}, pages 533--544. Springer, 2015.

\bibitem{FischerIKS18}
Johannes Fischer, Tomohiro I, Dominik K{\"{o}}ppl, and Kunihiko Sadakane.
\newblock Lempel-ziv factorization powered by space efficient suffix trees.
\newblock {\em Algorithmica}, 80(7):2048--2081, 2018.

\bibitem{Franklin82}
Wm.~Randolph Franklin.
\newblock On an improved algorithm for decentralized extrema finding in
  circular configurations of processors.
\newblock {\em Commun. {ACM}}, 25(5):336--337, 1982.

\bibitem{Frederickson87}
Greg~N. Frederickson.
\newblock Upper bounds for time-space trade-offs in sorting and selection.
\newblock {\em J. Comput. Syst. Sci.}, 34(1):19--26, 1987.

\bibitem{GawrychowskiK17}
Pawe{\l} Gawrychowski and Tomasz Kociumaka.
\newblock Sparse suffix tree construction in optimal time and space.
\newblock In {\em 28th {SODA}}, pages 425--439. {SIAM}, 2017.

\bibitem{Golynski09}
Alexander Golynski.
\newblock Cell probe lower bounds for succinct data structures.
\newblock In {\em 20th {SODA}}, pages 625--634. {SIAM}, 2009.

\bibitem{Guspiel}
Grzegorz Guśpiel.
\newblock An in-place, subquadratic algorithm for permutation inversion.
\newblock {\em CoRR}, abs/1901.01926, 2019.

\bibitem{HirschbergS80}
Daniel~S. Hirschberg and James~B. Sinclair.
\newblock Decentralized extrema-finding in circular configurations of
  processors.
\newblock {\em Commun. {ACM}}, 23(11):627--628, 1980.

\bibitem{ItaiR90}
Alon Itai and Michael Rodeh.
\newblock Symmetry breaking in distributed networks.
\newblock {\em Inf. Comput.}, 88(1):60--87, 1990.

\bibitem{Knuth71}
Donald~E. Knuth.
\newblock Mathematical analysis of algorithms.
\newblock In {\em {IFIP} Congress {(1)}}, pages 19--27, 1971.

\bibitem{KnuthVol1}
Donald~E. Knuth.
\newblock {\em The art of computer programming, Volume {I:} Fundamental
  Algorithms, 3rd Edition}.
\newblock Addison-Wesley, 1997.

\bibitem{MunroP80}
J.~Ian Munro and Mike Paterson.
\newblock Selection and sorting with limited storage.
\newblock {\em Theor. Comput. Sci.}, 12:315--323, 1980.

\bibitem{MunroR96}
J.~Ian Munro and Venkatesh Raman.
\newblock Selection from read-only memory and sorting with minimum data
  movement.
\newblock {\em Theor. Comput. Sci.}, 165(2):311--323, 1996.

\bibitem{PagterR98}
Jakob Pagter and Theis Rauhe.
\newblock Optimal time-space trade-offs for sorting.
\newblock In {\em 39th {FOCS}}, pages 264--268. {IEEE} Computer Society, 1998.

\bibitem{Peterson82}
Gary~L. Peterson.
\newblock An ${O}(n\log n)$ unidirectional algorithm for the circular extrema
  problem.
\newblock {\em {ACM} Trans. Program. Lang. Syst.}, 4(4):758--762, 1982.

\bibitem{RamanR99}
Venkatesh Raman and Sarnath Ramnath.
\newblock Improved upper bounds for time-space trade-offs for selection.
\newblock {\em Nord. J. Comput.}, 6(2):162--180, 1999.

\end{thebibliography}

\appendix

\section{Leader election by Fich et al.}\label{se:fich}

In this section, we describe an algorithm of cycle leader problem by Fich et al.
\cite{FichMP95}.
Compared to the original description, we introduce new notation and change some
details to make further algorithms easier to describe.
We stress that this should not be considered as new material and is provided only
for convenience.

Fich et al. \cite{FichMP95} consider sets of local minima on the cycles of
$\pi$. Let $E_1 = [n]$ and $\pi_1=\pi$.
For $r>1$, $E_r\subset E_{r-1}$ is defined inductively as the set of local
minima in $E_{r-1}$ following the order of permutation $\pi_{r-1}$, that is
$E_r=\{i\in E_{r-1} : \pi_{r-1}^{-1}(i)>i<\pi_{r-1}(i)\}$.
Similarly, $\pi_r : E_r \rightarrow E_r$ is the permutation that maps each
element of $E_r$ to the next element of $E_r$ that is encountered when following
the permutation $\pi_{r-1}$.
Note that at most every second element of $E_{r-1}$ can belong to $E_r$ and
hence $|E_r|\leq|E_{r-1}|/2$.
For a cycle $C$ of $\pi$, there exists $t \leq \lceil \log n \rceil$, such that
$E_{t+1} \cap C = \{x\}$, where $x$ is the minimum element on $C$. Of course,
$E_{t+2} \cap C = \emptyset$.

\FIGURE{h}{0.4}{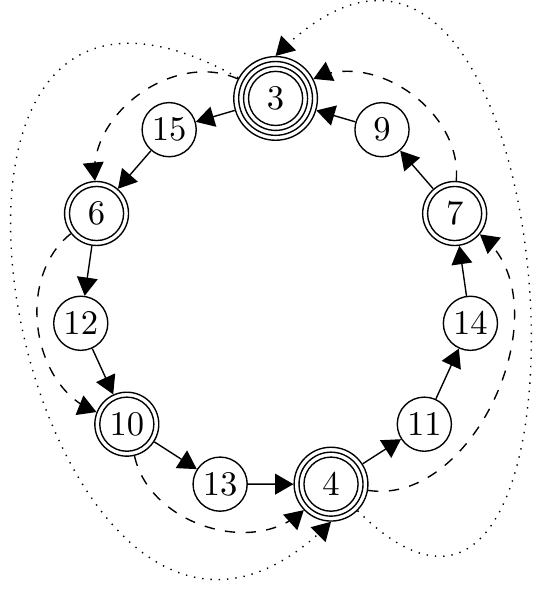}{
  Example of a single cycle of $\pi$. The number of circles around each vertex
  indicates the set $E_1$, $E_2$, $E_3$ or $E_4$ the vertex belongs to. Dashed
  edges denote $\pi_2$ and dotted edges denote $\pi_3$.
}

To simplify the presentation, we always restrict our analysis to a single
cycle. It generalizes easily to the original problem, as all presented
algorithms in this paper follow the pattern of Algorithm~\ref{alg:left_to_right}
and during execution of $\Process(i)$, we consider only the elements on the
cycle containing $i$.
Let us proceed even further and also assume that we are
dealing with paths. This will be needed in sections about inverting
permutations.

\begin{definition}
  A \emph{path} of $\pi$ is a~partial function that can be obtained
  from a~cycle $C$ of $\pi$ by replacing exactly one value of
  $\pi(x)$ to $\perp$ (for $x \in C$).
\end{definition}

On a~path, $\pi_r$ is undefined for the last element $\ell_r\in E_r$.
When deciding if an element is a local minimum on a~path, we disregard
comparisons with undefined elements and we assume that $\pi_r(\perp) = \perp$
for every $r$.

Because it is easy to compute $\pi_r(i)$ but difficult to compute
$\pi_r^{-1}(i)$, we define the leader of a cycle to be the element $i$ such that
$m = \pi_t(\pi_{t-1}(\dots(\pi_2(\pi_1(i)))\dots))\in E_{t+1}$ for $t$ such
that $|E_{t+1}| = 1$ and where $m$ is the minimum on the cycle.

In other words, to check if $i$ is the~leader, we need to inspect elements $i$,
$\pi_1(i)$, $\pi_2(\pi_1(i))$, $\dots$,
$m = \pi_t(\pi_{t-1}(\dots(\pi_2(\pi_1(i)))\dots))$ and check if they belong to
$E_1$, $E_2$, $E_3$, $\dots$, $E_{t+1}$ respectively.
Intuitively, the leader $i$ is the witness that $m$ is the minimum of its cycle.

We illustrate the definition on the example cycle in Figure
\ref{fig:fich-cycle} with $t=3$.
We~can quickly find that 13 is not the leader, because
$\pi_2(\pi_1(13)) = 7 \not \in E_3$.
The leader of the cycle is 12 as $\pi_1(12) = 10 \in E_2$,
$\pi_2(10) = 4 \in E_3$ and $\pi_3(4) = 3 \in E_4$.

\begin{definition}
 A \emph{staircase} of size $r$ from $i$ is a sequence of elements
 $(i=i_1,i_2,\ldots,i_{r+1}=m=j_{r+1},j_r,j_{r-1},\ldots,j_1=i')$ such that
 $i_k,j_k\in E_k$, for $k\in[r+1]$ and $i_{k+1}=\pi_k(i_k)$ and
 $j_k=\pi_k(j_{k+1})$ for $k\in[r]$.
 Elements $i$, $m$ and $i'$ are~called \emph{the start}, \emph{the middle}
 and \emph{the end} of the staircase, respectively.
 Part of the staircase from the start to the middle is called
 \emph{left part} of the staircase. Similarly, \emph{right part} of the
 staircase is the part from the middle to the end of the staircase.
\end{definition}

In the example from Figure~\ref{fig:fich-cycle}, the staircase of size $3$ from $12$ is the sequence
$(i = 12, 10, 4, m = 3, 4, 7, i' = 9)$.
See Figure~\ref{fig:straight-staircase} for a larger example.

The algorithm constructs staircases of the largest possible size from each
element.
It might happen that the constructed staircase is longer than a full cycle,
because it visits some elements twice
(once before and once after reaching the middle of the staircase). However, we
notice that all elements occur at most once in each part (left or right) of
the staircase.
This is the case for the staircase of size 3 from 12 in Figure~\ref{fig:fich-cycle}, where the part of cycle between 12 and 9 was visited twice.
As~we may have to work on a~path, we~avoid constructing staircases with
undefined elements. Of course, for paths, it is not possible to traverse any of its
elements more than once.

\FIGURE{h}{0.98}{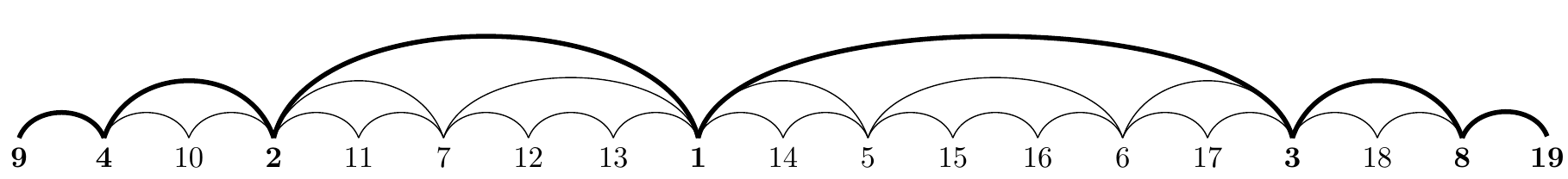}{
  A staircase of size 3 from $9$ to $19$ with the middle element $1$ is marked bold.
  Curves of increasing heights represent $\pi_1$,
  $\pi_2$ and $\pi_3$ respectively.
}

As staircases are crucial objects that we use in all the presented algorithms,
we summarize their definition and naming in the more concise form:
$$
\arraycolsep=2.2pt
\def\arraystretch{1.0}
\begin{array}{rcccccccccccccccl}
  \rin{E_1 ~~} & & \rin{E_2 ~~} & & & & ~ \rin{E_{r+1}} & & \rin{E_r ~~} & & \rin{E_{r-1}} & & & & \rin{E_2 ~~} & & \rin{E_1 ~~} \\
  i=i_1 & \myarrow{1} & i_2 & \myarrow{2} & \cdots & \myarrow{r} & i_{r+1}=m=j_{r+1} & \myarrow{r} & j_r & \myarrow{r-1} & j_{r-1} & \myarrow{r-2} & \cdots & \myarrow{2} & j_2 & \myarrow{1} & j_1 = i' \\
  \textrm{start} & & & & & & \textrm{middle} & & & & & & & & & & \textrm{end} \\
\end{array}
$$

The heart of the algorithm of Fich et al. is the procedure $\Next(r)$, which computes the successive
elements of $\pi_r$ recursively, using the successive elements of $\pi_{r-1}$.
It operates on table $\elbow$ consisting of elements from the cycle and changes it during the execution.
The algorithm is presented in Algorithm~\ref{alg:next}, with the only our modification in Line~\ref{li:abort}, which aborts when attempting to pass the end of the path.
Before we describe the procedure in detail, we start with presenting the properties of the table $\elbow$ that are satisfied before and after execution of $\Next(r)$.

\begin{algorithm}[t]
\begin{algorithmic}[1]
  \Function{$\Next$}{$r$}
  \If{$r=1$}
    \If{$\pi(\elbow[1]) = \perp$}\label{li:abort}
      \textbf{abort}
    \EndIf
    \State $\elbow[0] \gets \pi(\elbow[1])$
    \State \Return
  \EndIf
  \While{$\elbow[r-1]<\elbow[r-2]$}
    \State $\elbow[r-1] \gets \elbow[r-2]$
    \State $\Next(r-1)$
  \EndWhile
  \While{$\elbow[r-1]>\elbow[r-2]$}
    \State $\elbow[r-1] \gets \elbow[r-2]$
    \State $\Next(r-1)$
  \EndWhile
 \EndFunction
\end{algorithmic}
\caption{A recursive procedure that rearranges the $\elbow$ array and computes
$\pi_r(\elbow[r])$.}
\label{alg:next}
\end{algorithm}

The procedure operates on table \textit{elbow} which prior to calling
$\Next(r)$ satisfies:
\begin{equation} \label{eqn:fich_assumption}
  E_r\ni\elbow[r]=\elbow[r-1]\myarrow{r-1}\elbow[r-2]\myarrow{r-2}\cdots\myarrow{2}\elbow[1]\myarrow{1}\elbow[0]
\end{equation}
which is a shorthand for
$\elbow[r-1]=\elbow[r]\in E_r$ and $\elbow[k-1]=\pi_k(\elbow[k])$ for
$1\leq k\leq r-1$.
Immediately after termination of $\Next(r)$ the table $\elbow$ satisfies:
\begin{equation}\label{eqn:fich_end_assumption}
  E_r\ni\elbow[r]\myarrow{r}\elbow[r-1]\myarrow{r-1}\elbow[r-2]\myarrow{r-2}\cdots\myarrow{2}\elbow[1]\myarrow{1}\elbow[0]
\end{equation}

Notice that only the element $\elbow[r]$ remains unchanged during the execution
of $\Next(r)$.
The main output of the function, $\pi_r(\elbow[r])$, is stored in $\elbow[r-1]$.
Observe that the $\elbow$ array contains exactly the right part of the staircase of size $r$ with the middle element $\elbow[r]$: $\elbow[k] = j_{k+1}$ for $k \in [r]$.
We use 0-based indexing for the $\elbow$ array to be consistent with the algorithm of Fich et al.

Due to the condition \eqref{eqn:fich_assumption}, $\elbow[r]\in E_r$
and hence all the subsequent elements of $\elbow$ satisfying the condition
\eqref{eqn:fich_end_assumption} are well-defined, as $\elbow[k]\in E_{k+1}\subset E_k$ for $k\in[r-1]$.
Furthermore, it is not possible to calculate the values in the left part of the
staircase based only on the $\elbow$ array, so we need to ensure that upon calling
$\Next(r)$, condition \eqref{eqn:fich_assumption} is satisfied.

In the first while loop in Algorithm~\ref{alg:next}, the procedure iterates over $\pi_{r-1}$ as long as the
corresponding values on the cycle are increasing.
As the cycle is nontrivial, the loop terminates when $\elbow[r-1]>\elbow[r-2]$.
Similarly, the second while loop iterates over decreasing values on $E_{r-1}$ and
terminates when $\elbow[r-1]<\elbow[r-2]$. Then $\elbow[r-1]$ is the next local
minimum on $\pi_{r-1}$ after $\elbow[r]$, and hence
$\elbow[r-1]=\pi_r(\elbow[r])$.

Observe that elbows create a nested
structure, that is for each $m$ and $r$ such that $m\in E_{r+1}$, the element
$\pi_1(\pi_{2}(\dots(\pi_r(m))\dots))$ appears before $\pi_{r+1}(m)$ on
the cycle. This follows by induction. Similarly, we can say that
$\pi_r(\pi_{r-1}(\dots(\pi_1(i))\dots))$ appears earlier than $\pi_{r+1}(i)$ on
the cycle.
See Figure~\ref{fig:recursive-elbows} for an example.

\FIGURE{h}{0.98}{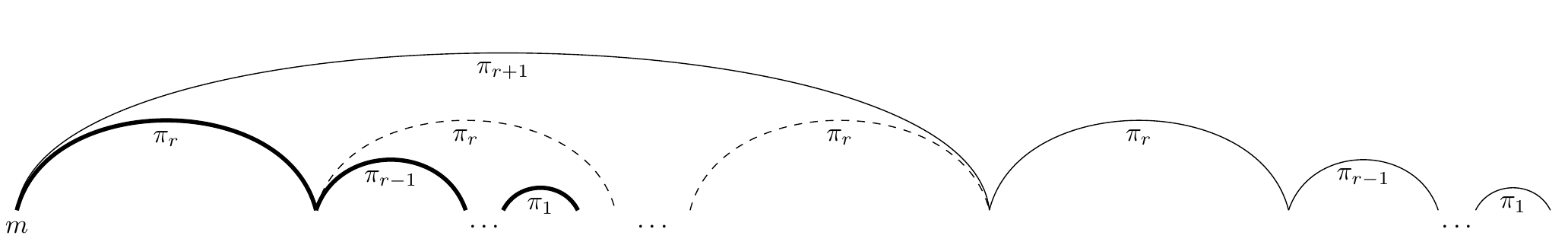}{Nested structure of elbows. Each
$\pi_{r+1}(m)$ is $\pi_r^p(m)$ for some $p \ge 2$ and appears
after $\pi_{1}(\pi_{2}(\dots(\pi_r(m))\dots))$ when traversing from $m$.
The entries of the $\elbow$ table before(after) the execution of $\Next(r+1)$
procedure are marked with a bold(solid) line.
}

Now we introduce more definitions and properties of the staircases.

\begin{definition}
  An \emph{almost staircase} of size $r$ from $i$ is a sequence of elements
  $(i=i_1,i_2,\ldots,i_{r+1}=m=j_{r+1},j_r,j_{r-1},\ldots,j_1=i')$, such that
  $i_k,j_k\in E_k$, for $k\in[r]$ and $i_{k+1}=\pi_k(i_k)$ and
  $j_k=\pi_k(j_{k+1})$ for $k\in[r]$.
\end{definition}

Notice the difference between staircase and almost staircase is that in the
latter we do not require its middle element $m$ to belong to $E_{r+1}$.
This might happen due to $m=\pi_r(i_r)$ implying that $m \in E_r$, but not
necessarily $m \in E_{r+1}$.

\begin{definition}
  An almost staircase
  $(i=i_1,i_2,\ldots,i_{r+1}=m=j_{r+1},j_r,j_{r-1},\ldots,j_1=i')$
  is called \emph{proper} if $i_r \neq m$ (which is equivalent to $m \neq j_r$).
\end{definition}

On a cycle, it is always possible to extend a given staircase of size $r$ from
$i$ with the middle $m$, to an~almost staircase of size $r+1$ from $i$.
If the extension is not proper, then $i_r = m$, so $m$ is the only
element on $E_r$. In this case, the staircase is the largest
possible from $i$ and $m$ is the minimum of the cycle.

On a path, it might be not possible to extend a given staircase of size $r$ from
$i$ with the middle $m$ to an~almost staircase from $i$ of size $r+1$, because
of reaching $\perp$.
If the extension to the almost staircase of size $r+1$ from $i$ was not
possible, the staircase of size $r$ is the largest possible from $i$.
Notice that every almost staircase on a path is proper.

\begin{definition}\label{def:BestStaircase}
  Consider a~staircase from $i$ of size $t$.
  We call it the \emph{best staircase} from $i$ if there is no proper almost
  staircase of size $t+1$ from $i$.
\end{definition}

In the example cycle from Figure \ref{fig:fich-cycle}, it is possible
to construct the best staircase from 12: $(12, 10, 4, 3, 4, 7, 9)$,
but the staircase $(14, 7, 3, 6, 12)$ is not a~best staircase,
as it is possible to construct a proper almost staircase
$(14, 7, 3, 4, 3, 6, 12)$, which is not a~staircase.

\begin{algorithm}[t]
\begin{algorithmic}[1]
  \Function{$\BestStaircase$}{$i$}
    \State $\elbow[0] \gets \elbow[1] \gets i$
    \For{$r=1..$}
      \State $m \gets \elbow[r]$
      \If{\Next($r$) aborted}\label{line:check-aborted}
        \State \Return $m$ \Comment{Best staircase of the path is found.}
      \EndIf
      \State $m' \gets \elbow[r] \gets \elbow[r-1]$
      \If{\Next($r$) aborted}\label{line:check-aborted2}
        \State \Return $m$\Comment{Best staircase of the path is found.}
      \EndIf
      \State $m'' \gets \elbow[r-1]$\label{line:best-staircase-half-step}
      \If{$m' = m$}
        \State \Return $m$ \Comment{Best staircase of the cycle is found.}
        \label{line:best-staircase-cycle-returned}
      \EndIf
      \If{$m' < m$\textbf{ and }$m' < m''$}\label{line:best-check}
        \State $\elbow[r+1] \gets \elbow[r]$
        \Comment{Constructed proper almost staircase is a staircase.}
      \Else
        \State \Return \NULL
        \Comment{Constructed proper almost staircase is not a staircase. 
        }
      \EndIf
    \EndFor
  \EndFunction
\end{algorithmic}
\caption{Constructs the best staircase from $i$ (if any).}
\label{alg:BestStaircase}
\end{algorithm}

Algorithm~\ref{alg:BestStaircase} constructs the best staircase from $i$ or
returns $\NULL$ if there is no best staircase from $i$.
It constructs almost staircases of increasing sizes in subsequent
iterations, checks if these are staircases (line \ref{line:best-check})
and if so, updates the $\elbow$ table. The algorithm succeeds upon finding that
there are no larger almost staircases, in such a case the middle of the staircase
is returned.
The checks in line \ref{line:check-aborted} and \ref{line:check-aborted2} are
our modifications of original algorithm to make it working also for paths.

For example, procedure $\BestStaircase(9)$ executed on the part of the cycle from Figure
\ref{fig:straight-staircase}, constructs the following
staircases in the subsequent iterations of the main loop: $(9)$, $(9, 4, 10)$,
$(9, 4, 2, 7, 12)$ and $(9, 4, 2, 1, 3, 8, 19)$.

\begin{lemma}
  For every~cycle $C$ of permutation $\pi$ there is exactly one element
  $x \in C$, from which there is a~best staircase.
  \label{lemma:unique_cycle_leader_fich}
\end{lemma}
\begin{proof}
  On cycles there is no $\perp$, so $\Next$ will never abort.
  This means that the only way for
  Algorithm~\ref{alg:BestStaircase} to return the best staircase is in line
  \ref{line:best-staircase-cycle-returned}, when the constructed almost
  staircase from $i$ is~not~proper, so $m$ is the only element on $E_{t+1}$,
  thus the minimum on the cycle.
  Indeed, there is the~best staircase from
  $i = \pi^{-1}_1(\pi^{-1}_2(\dots(\pi^{-1}_t(m))\dots))$.

  Notice that on a cycle, every best staircase has middle in $m$.
  There cannot be two best staircases with the same middle $m$ and different
  sizes, as the smaller could have been extended to a proper almost staircase
  and thus is not best.
  This concludes the proof of the lemma.
\end{proof}

\begin{definition}
 An element $i$ is the leader of its cycle if there exists the best staircase
 from~$i$.
 \label{def:leader}
\end{definition}

There can be more than one element on~a~path from which there is a best
staircase.
We will define leaders for paths in Appendix \ref{se:inverting-b1}.

Intuitively, our notion of a leader allows us to elect it efficiently because we can
terminate $\BestStaircase(i)$ as soon as we encounter the first $r$
such that $\pi_{r}(\pi_{r-1}(\dots(\pi_1(i))\dots)) \not \in E_r$.
Using this definition, we can simply check if there is the best staircase from
$i$ and if so, return $i$ as the leader of the cycle, see
Algorithm~\ref{alg:process_with_best}.

\begin{algorithm}[t]
\begin{algorithmic}[1]
  \Function{$\Process$}{$i$}
    \State $m \gets \BestStaircase(i)$
    \If{$m \ne \NULL$}
      \State Report that $i$ is the leader
    \EndIf
 \EndFunction
\end{algorithmic}
\caption{Checks if $i$ is the leader of its cycle.}
\label{alg:process_with_best}
\end{algorithm}

\begin{theorem}\label{thm:fich}
  There exists an algorithm for reporting leaders of a permutation $\pi$ on $n$
  elements in $\Oh(n \log n)$ time using $\Oh(\log^2 n)$ additional bits of space.
\end{theorem}
\begin{proof}
  The correctness of the algorithm follows immediately from
  Lemma~\ref{lemma:unique_cycle_leader_fich}.
  The cost of~$\BestStaircase$ is proportional to the total number of calls to
  the $\Next$ function, which is asymptotically the number of accesses to the
  input permutation $\pi$.
  Observe that during the execution of $\BestStaircase(i)$ for a particular $i$, every call to $\Next(1)$ that accesses the permutation $\pi$,
  proceeds along the cycle, so the execution of $\Process(i)$ and testing if $i$ is
  the leader involves $k$ steps along the cycle, from $i$ to $\pi^k(i)$ for some
  $k$. Now we need to bound the total number of steps performed for all starting
  elements $i$.

  \FIGURE{h}{0.98}{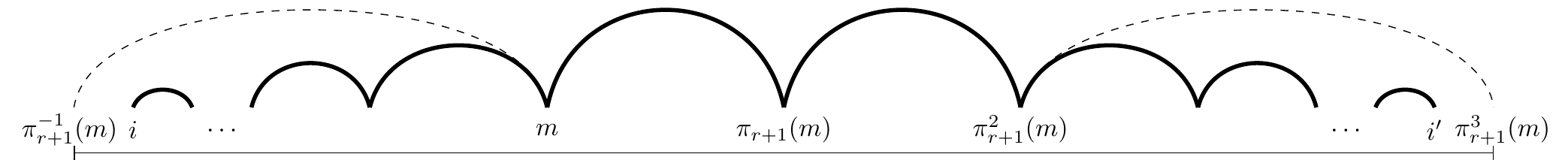}{An~almost staircase of size $r+1$ from $i$.
  $m$ is the middle element of a~staircase of size $r$ from $i$.}

  Consider an element $i$.
  Suppose the largest staircase created during execution of $\BestStaircase(i)$ have size $r$ and middle $m$. We do not distinguish the reason why there is no staircase of larger size.
  From the nested structure of elbows, $i$ is between $\pi_{r+1}^{-1}(m)$ and
  $m$ on the cycle and construction of a larger almost staircase terminated between $m$ and
  $\pi_{r+1}^3(m)$. See Figure \ref{fig:range-of-middle}.
  Hence the traversed part of the cycle while running $\Process(i)$ is
  bounded by the part of the cycle between $\pi_{r+1}^{-1}(m)$ and
  $\pi_{r+1}^3(m)$ and we say that $i$ charges the time of $\BestStaircase(i)$ to $m$ at level $r$.
  As for a fixed element $m$ and level $r$ there is at most one $i$ charging to $m$ at level $r$, the total time for processing all elements is bounded by:
  $$
    \sum_{r = 1}^{\lceil \log n \rceil} \sum_{m \in E_r}
    \left(
      \len{\pi_{r+1}^{-1}(m)}{m}
    + \len{m}{\pi_{r+1}(m)}
    + \len{\pi_{r+1}(m)}{\pi_{r+1}^2(m)}
    + \len{\pi_{r+1}^2(m)}{\pi_{r+1}^3(m)}
    \right).
  $$
  where we split interval $\pi_{r+1}^{-1}(m)\cdots \pi_{r+1}^3(m)$ into four terms in order to cover also the case when $|E_{r+1}|<4$.
  As the inner sum is bounded by $4n$, the overall running time is $\Oh(n \log n)$.
  
  The list of elbows consists of $\Oh(\log n)$ words of $\Oh(\log n)$ bits and
  the depth of the recursion is $\Oh(\log n)$, so~the space complexity
  is $\Oh(\log^2 n)$ bits.
\end{proof}

\section{Inverting permutation}\label{se:inverting-b1}

To invert a~cycle of~a~permutation, we use the following procedure of El-Zein et al. \cite{El-ZeinMR16}:
\begin{algorithm}[t]
\begin{algorithmic}[1]
  \Function{$\InvertCycle$}{$i$}
    \State $p \gets i$
    \State $x \gets \pi(i)$
    \While{$x \neq i$}
      \State $n \gets \pi(x)$
      \State $\pi(x) \gets p$ \label{li:update_pi}
      \State $p \gets x$
      \State $x \gets n$
    \EndWhile
    \State $\pi(x) \gets p$
  \EndFunction
\end{algorithmic}
\caption{Invert the~cycle with element $i$.}
\label{alg:invert_cycle}
\end{algorithm}

In each iteration of the while loop in Algorithm~\ref{alg:invert_cycle}, $x$ is the
position to update and before executing line \ref{li:update_pi} of the
algorithm we have $p = \pi^{-1}(x)$ and $n = \pi(x)$.
We obtain $\Oh(1)$-space, $\Oh(n^2)$-time algorithm which computes
$\pi^{-1}$ from $\pi$ by combining
Algorithms~\ref{alg:process}~and~\ref{alg:invert_cycle}:
instead of reporting leaders, one should just invert the~cycle.

The natural idea is to invert each cycle upon reaching its leader. This works
for the algorithm electing as leader the minimum of the cycle, as the leader does not change after inverting the cycle.
Unfortunately, the definition of leader given in previous section is sensitive
to cycle inversion.
This is exactly why simply executing Algorithm~\ref{alg:process_with_best}
instead of Algorithm~\ref{alg:process} is not a~solution, because we need to
ensure that we do not invert a cycle twice, but we cannot mark that a cycle
is already inverted.

More precisely, consider a cycle $C$ with the leader $i$
and the minimum element $m$. When Algorithm~\ref{alg:process_with_best} inverts
$C$, it might happen that the leader $j$ of the inverted cycle is greater than
$i$ and in this case, the already inverted cycle $C$ gets inverted again while
considering $j$. This explains why our algorithm (and the algorithm of
Fich et al.) cannot be applied directly for inverting permutations.

To overcome this difficulty, our algorithm (similar to~\cite{El-ZeinMR16}) will 
mark such difficult cycles by converting them to paths. Later,
before terminating the whole algorithm, each path has to be repaired to a~cycle.
At a high level, this is a~simplification of \cite{El-ZeinMR16}, because instead of having four
different scans running simultaneously, we have only one, corresponding to $\mathscr{F}$ there.
Moreover, in our algorithm, because we store extended ranks, each hard cycle is
fixed only once, exactly when needed. This makes both the algorithm and analysis
simpler.

Observe that using the approach from the previous section, we can compute
$j$ as the end of the largest staircase constructed from $i$ which
we can read from $\elbow[0]$. Indeed, after inverting a~cycle, the middle of
the best staircase remains the same, so the best staircase of the inverted cycle
is  $(i', j_2, j_3, \dots, j_r, m, i_r, \dots, i_3, i_2, i)$, exactly the
reverse of the previous best staircase.

\begin{definition}
  A~cycle $C$ of $\pi$ is~\emph{hard} if~its best staircase
  $(i, i_2, i_3, \dots, i_r, m, j_r, \dots, j_3, j_2, i')$ satisfies $i' > i$.
  Otherwise, $C$ is \emph{easy}.
\end{definition}

At a high level, to compute $\pi^{-1}$ we invert cycles when processing their leaders.
If the cycle is hard we need to avoid inverting it again while processing the leader of the inverted cycle.
For this purpose we replace one value of $\pi$ on the cycle with $\perp$ and the cycle becomes a path.
The path will be transformed to the cycle again while processing the leader of the inverted cycle.
Then the cycle will be never modified again.

On a path, while processing an element $i$, we are unable to visit its preceding elements and so we cannot check if an element is minimum on the path as in Algorithm~\ref{alg:process_with_best}.
In order to be able to recognise this situation, we store some additional information in each $\perp$.

Recall that in our model, we can only replace $\pi(i)$
with some value from $[n]$.
El-Zein et al. \cite{El-ZeinMR16} introduce an approach which, using some
additional space, allows to virtually extend the set of possible
values stored in $\pi$, so that we can design algorithms in which
$\pi(i)\in [n] \cup \{\perp_1,\perp_2,\ldots,\perp_k\}$.
$\perp_i$ is the $i$-th \emph{null} character and it is possible to retrieve the
value of $i$ from $\perp_i$.
In our application, each value from $[n]$ appears at most $c$ times in
$\pi$, but for any $i$ there can be many elements $j$ such that $\pi(j)=\perp_i$.

\begin{lemma}[\cite{El-ZeinMR16}]\label{le:nulls}
 For any $k\in [n]$, it is possible to simulate an extension of the range of stored values of
 $\pi$ from $[n]$ to $[n] \cup \{\perp_1,\ldots,\perp_k\}$ with each value from
 $[n]$ occurring at~most $c$ times in $\pi$ using $\Oh(ck\log n)$ bits with
 $\Oh(c)$-time overhead.
\end{lemma}
\begin{proof}
  We store $\pi$ as an array $A$ with values from $[n]$ of size $n$ and an
  array $B$ of $k$ lists, each of length at most $c$.
  The $k$-th list $B[k]$ stores all values $i$ such that $\pi(i)=x$.
  We update the value of $\pi$ in the following way:

  \begin{itemize}
    \item for the update $\pi(i) \leftarrow x \in [n]$, we set $A[i] \gets x$ and if $x \le k$
    insert $i$ into the list $B[x]$,
    \item for the update $\pi(i) \leftarrow \perp_x$, we set $A[i] \gets x$.
  \end{itemize}

  If before the update held $\pi(i)=x\in[k]$, we first remove $i$ from $B[x]$.
  To retrieve the value of $\pi(i)$, we return $A[i]$ if it is larger than $k$.
  Otherwise, $A[i] = x \le k$ and $\pi(i)$ is either $x$ or $\perp_x$. To distinguish
  between these cases, we check the list $B[x]$ and return
  $x$ if $i \in B[x]$ or $\perp_x$ otherwise.
\end{proof}
\noindent
When the exact value stored in $\perp$ is not relevant, we will use the symbol
$\perp$ denoting $\perp_k$ for any~$k$.

\begin{definition}
 A \emph{half staircase} of size $r$ from $i$ is a sequence
 of elements $(i=i_1,i_2,\ldots,i_{r+1}=m=j_r,j_{r-1},j_{r-2},\ldots,j_1=i')$
 such that $i_k,j_\ell\in E_k$, for $k\in[r]$, $\ell\in[r-1]$ and $i_{k+1}=\pi_k(i_k)$
 for $k\in[r]$ and $j_\ell=\pi_\ell(j_{\ell+1})$ for $\ell\in[r-1]$.
\end{definition}
Intuitively, a~half staircase of size $r+1$ is the intermediate state
between the~staircase of size $r$ and an almost staircase of size $r+1$.
See Figure \ref{fig:life-of-staircase} for reference.

\FIGURE{h}{0.98}{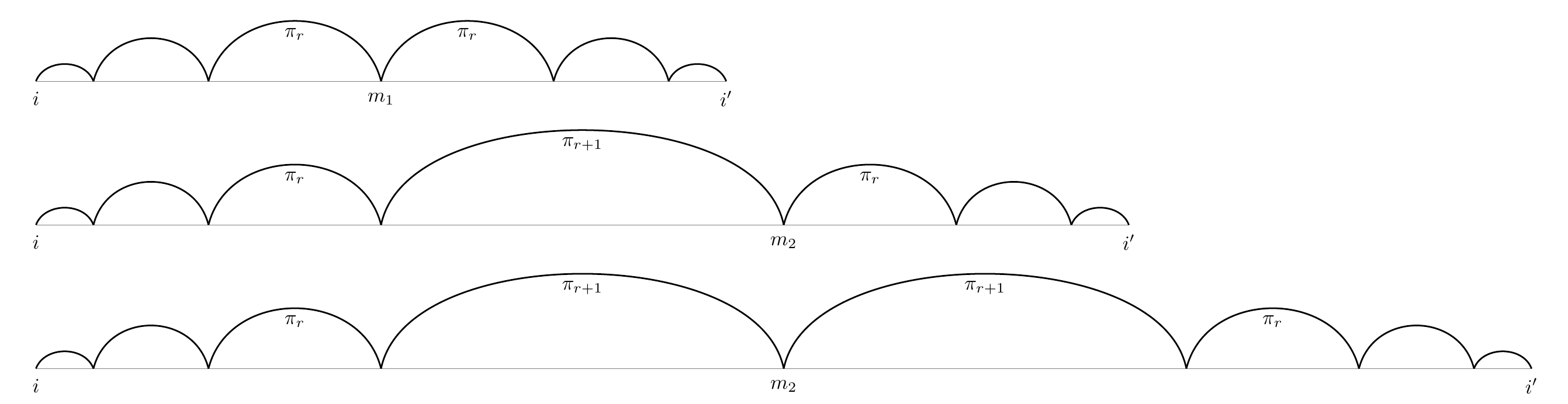}{
  A~staircase of size $r$, a~half staircase of size $r+1$ and an~almost
  staircase of size $r+1$.
}

Now we modify procedure $\BestStaircase(i)$ to return a~tuple $(m, i', t, h)$ where
$m$ is the middle of~the constructed best staircase from $i$, $i'$ is~its end,
$t$ is its size and $h$ is a~Boolean
value denoting if an abort occurred after constructing a half staircase of
size $t+1$ from $i$.
\begin{algorithm}[t]
\begin{algorithmic}[1]
  \Function{$\BestStaircase$}{$i$}
    \State $\elbow[0] \gets \elbow[1] \gets i$
    \For{$r=1..$}
      \State $pi' \gets \elbow[0]$
      \State $m \gets \elbow[r]$
      \If{\Next($r$) aborted} \label{li:returnBS2-1}
        \State \Return $(m, pi', r-1, \false)$
      \EndIf
      \State $m' \gets \elbow[r] \gets \elbow[r-1]$ \label{li:half-constructed}
      \If{\Next($r$) aborted} \label{li:returnBS2-2}
        \State \Return $(m, pi', r-1, \true)$
      \EndIf
      \State $m'' \gets \elbow[r-1]$ \label{li:partial-constructed}
      \If{$m = m'$} \label{li:returnBS2-3}
        \State \Return $(m, pi', r-1, \true)$ 
      \EndIf
      \If{$m' < m$\textbf{ and }$m' < m''$}
        \State $\elbow[r+1] \gets \elbow[r]$
      \Else
        \State \Return \NULL
      \EndIf
    \EndFor
  \EndFunction
\end{algorithmic}
\caption{Modified function that constructs the best staircase from $i$ (if any).}
\label{alg:BestStaircase2}
\end{algorithm}

Algorithm~\ref{alg:BestStaircase2} implements the lifecycle of staircase extension.
We try to extend the staircase by constructing a half staircase of bigger size, then almost staircase and then checking if the almost staircase is a staircase.
See Figure \ref{fig:life-of-staircase}.
More precisely, the $r$-th iteration of the main loop starts with the already constructed staircase
of size $r-1$. In line \ref{li:half-constructed}, a half staircase of size $r$
is constructed and then, in line \ref{li:partial-constructed} the algorithm
extends it to an almost staircase of size $r$. If an abort occurred while constructing an almost staircase of size $r$ (in line \ref{li:returnBS2-1} or \ref{li:returnBS2-2}) or if the constructed almost staircase is
not proper (line \ref{li:returnBS2-3}), then the~staircase of size $r-1$ from the beginning of
iteration is the~best staircase from $i$.
If a proper~almost staircase of size $r$ is~successfully created, the algorithm
checks if it is a~staircase of size $r$ and either proceeds to next iteration or
returns $\NULL$.

\begin{definition}
  The \emph{rank} of~element $i$ on~a~path of $\pi$ is~the size of the best
  staircase from $i$. If there is no~best staircase from $i$ then the rank of
  $i$ is~undefined.
\end{definition}

\begin{definition}
  The \emph{extended rank} of~element $i$ on~a~path of $\pi$ is~pair $(t, h)$,
  where $t$ is the rank of $i$ and $h$ denotes if it
  is possible to construct the half staircase of size $t + 1$ from $i$. If $t$
  is undefined then the extended rank of $i$ is also undefined.
\end{definition}

Extended ranks are compared lexicographically. The concept
of rank is similar to \emph{limited depth} in \cite{El-ZeinMR16}. Our
contribution is the idea of half staircases and extended ranks. To see why it
makes further considerations easier, we show the following fact:

\begin{lemma}
  For any path $P$ of $\pi$ there is a~unique element of $P$ with the
  maximum extended rank.
  \label{lemma:unique-leader}
\end{lemma}
\begin{proof}
  First, observe that the extended rank $(t,h)$ and the middle $m\in E_{t+1}$
  uniquely determines the start of the (half)
  staircase, which is $\pi_1^{-1}(\pi_2^{-1}(\dots(\pi_t^{-1}(m))\dots))$ if
  $h = \False$ and $\pi_1^{-1}(\pi_2^{-1}(\dots(\pi_{t+1}^{-1}(m))\dots))$ if
  $h = \True$. See Figure \ref{fig:life-of-staircase}.

  Suppose that there are two distinct elements $i$ and $j$ with the same
  maximum extended rank $(t,h)$ and let $x$ be an element of $\pi$ such that $\pi(x) =\perp$.
  Without loss of generality, assume that $\len{i}{x} > \len{j}{x}$, that is $i$ is before $j$ on the path.
  Let $m_i$ and $m_j$ be the middle of the best staircase of size $t$ starting from $i$ and $j$, respectively.
  By definition, both $m_i$ and $m_j$ belong to $E_{t+1}$.
  From the nested structure of staircases, $\pi_{t+1}(m_i)$ is not later than $m_j$ on the path.
  Consider first the case when $h=\False$.
  Then it is possible to construct a half staircase of size $t+1$ from $i$, because
  $\pi_1(\pi_2(\ldots(\pi_{t+1}(m_i))\ldots))$ is not later than $\pi_1(\pi_2(\ldots(\pi_t(m_j))\ldots))$ on the path and thus we can increase the extended rank of $i$.
  For $h=\True$, in a similar way we can construct an almost staircase of size $t+1$ from $i$ which contradicts that there is a best staircase of size $t$ from $i$.
  See Figure~\ref{fig:proof-staircases}.
\end{proof}

\FIGURE{h}{0.98}{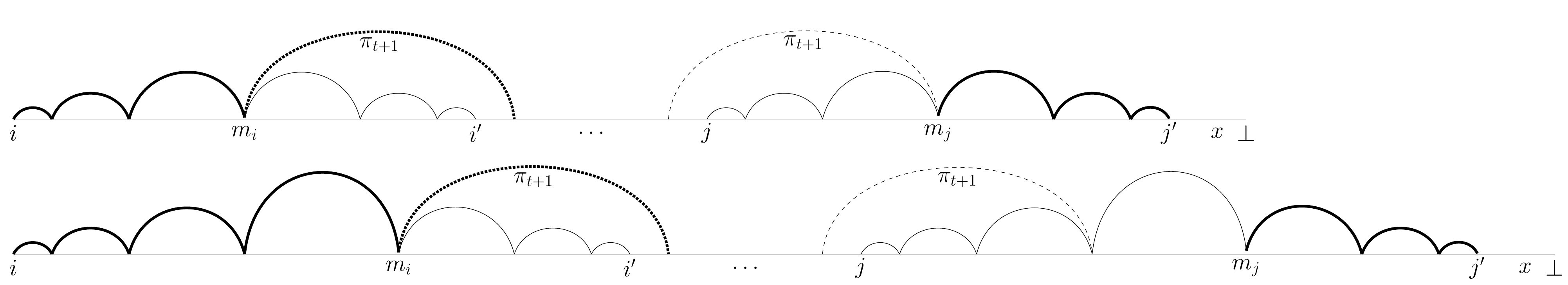}{
    Two cases from the proof: on the top of the picture $i$ and $j$ have the extended
    rank $(t,0)$ and on the bottom they have extended rank $(t,1)$.
    For the first case, we show that $i$ has the extended rank $(t,1)$, because
    the~half staircase of size $t+1$ from $i$ has its end not later
    than $j'$.
    Similarly, for the second case, we show that $i$ has extended rank
    $(t+1,0)$.
    In the figure, it is assumed that $j$ is later on the path than $m_i$
    (which actually is always true), but the proof does not use this assumption.
  }

In fact, by a~simple induction and applying the same reasoning as in the proof
of Lemma~\ref{lemma:unique-leader} in every inductive step, we can prove even
stronger property:
\begin{observation}
  Let $(e_1, e_2, \dots, e_k)$ be all the elements on a~path $P$ of $\pi$
  with defined extended rank, listed in the natural order. The sequence of
  their corresponding extended ranks is strictly decreasing.
  \label{observation:ranks-increasing-along-path}
\end{observation}

We note that from Lemma~\ref{lemma:unique-leader} it follows that on
every path there are at most two elements of the maximum rank and this fact was also used
in \cite{El-ZeinMR16}.

\begin{definition}
  The~\emph{leader} of path $P$ of $\pi$ is the element of $P$ with the largest
  extended rank.
\end{definition}

Clearly, by Lemma~\ref{lemma:unique-leader}, there is a unique leader on every path.
For a~cycle $C$ and any $x \in C$ and $y = \pi(x)$, by setting $\pi(x) \gets \perp$
we obtain a~path $P$ with $y$ as its first element and $x$ as the last
element. We call this operation \emph{cutting before $y$}.

\begin{proposition}
  If a~cycle $C$ of $\pi$ with best staircase from $i$ is cut before $i$,
  then the~best staircase from $i$ in $P$ exists.
  \label{prop:cut-is-safe}
\end{proposition}
\begin{proof}
  Let $S_1$ be the sequence of staircases created during construction of the
  best staircase from $i$ in $C$
  and let $S_2$ be the sequence of staircases created during construction of the
  best staircase from $i$ in $P$.
  $S_2$ is a~prefix of $S_1$, so the construction of best staircase from $i$ in
  $P$ never fails due to the fact that proper almost staircase is not a~staircase,
  but only because of an abort during the construction of almost staircase of
  larger size.
\end{proof}

Notice that if~a~cycle with leader $i$ is~cut before $i$
then, by Proposition~\ref{prop:cut-is-safe}, the rank of $i$ is still defined,
although it can be less than the size of the best staircase from $i$ before
cutting the cycle. From Observation
\ref{observation:ranks-increasing-along-path} we obtain:

\begin{corollary}
  For a~cycle $C$ and the best staircase from $i$ ending in $i'$, after
  inverting the cycle and cutting it before $i'$, we obtain a path $P$
  starting with $i'$ being its leader.
\end{corollary}

We now describe the new procedure $\Process$ for inverting
permutations. During the execution of the algorithm, each cycle of the
permutation can be in one of three following states:
\begin{itemize}
  \item the leader of the cycle has not been processed yet,
  \item the leader of the cycle has been processed, the cycle is inverted, but
  the leader of the inverted cycle has not been processed,
  \item the cycle is inverted and the leader of the inverted cycle has been already processed.
\end{itemize}
The second case applies only to hard cycles.
In this case, after inverting the cycle we cut it before $i'$ where $i'$ is the
leader of the inverted cycle.

\FIGURE{h}{0.98}{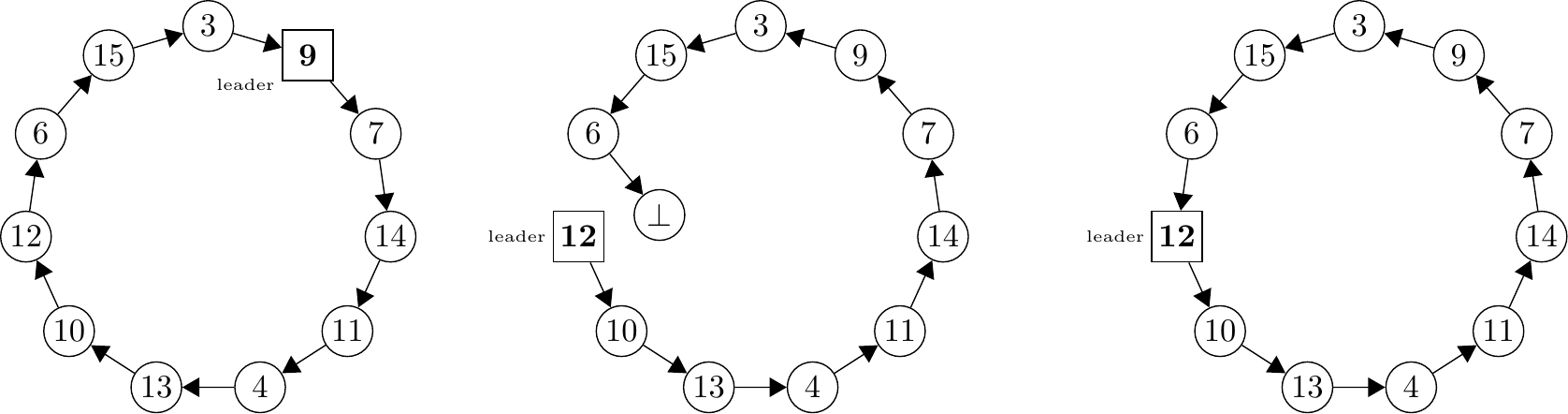}{
  Lifecycle of a hard cycle: before execution of $\Process(9)$, after
  $\Process(9)$ but before execution of $\Process(12)$ and the final cycle
  after execution of $\Process(12)$, respectively. Leaders of cycles/paths are
  denoted with squares.
}

\begin{algorithm}[t]
\begin{algorithmic}[1]
  \Function{$\Process$}{$i$}
    \State $(m, i', r, h) \gets \BestStaircase(i)$
    \If{$m \ne \NULL$}
      \State $a \gets \PathEnd(i)$
      \If{$a = \NULL$ \textbf{ and } $\Min(i,i)= m$}
        \State $b \gets \pi(i')$
        \State $\InvertCycle(i)$
        \If{$i' > i$}
          \State $\pi(b) \gets \perp_?$
          \Comment{The value in the null is yet to be determined.}
          \State $(\_, \_, r', h') \gets \BestStaircase(i')$
          \State $\pi(b) \gets \perp_{2r' + h'}$
          \Comment{Actual extended rank on the path is computed.}
        \EndIf
      \EndIf
      \If{$a \in [n] \textbf{ and } \pi(a) = \perp_{2r + h}$}
        \State $\pi(a) \gets i$
      \EndIf
    \EndIf
  \EndFunction
\end{algorithmic}
\caption{New version of the function to process an element in order to
invert a permutation.}
\label{alg:process_with_best_inverse}
\end{algorithm}

The final approach for processing $i$ is described in
Algorithm~\ref{alg:process_with_best_inverse} where
$\PathEnd$ is a~procedure which naively traverses the cycle from $i$ and
either returns $\NULL$ if the cycle was fully traversed or returns an element $a$,
such that $\pi(a) = \perp$. If a Boolean value appears in any arithmetic
operation, $\false$ is evaluated to $0$ and $\true$ to $1$.

$\Process(i)$ first checks if a~best staircase can be constructed from $i$.
If not, the execution for $i$ ends.
Let $i'$ be the end of the constructed best staircase.
If the algorithm is processing a cycle, we invert it.
If the cycle is hard, we cut the inverted cycle before $i'$,
compute the extended rank of $i'$ on the path
and then store it in the value of null for the predecessor of
$i'$ in the inverted cycle, which is now the end of the path.
Later, for the path, there might be multiple elements from which we can create
best staircases, but $i'$ is the only element with the largest extended rank.
For every element with the best staircase, we compare its extended rank with the
one stored in null.
If they are equal, this means that we found the leader of the path and we fix
the path into the~cycle by setting $\pi(a)=i$, where $a$ is end of the path.
Then, the currently processed element $i$ is the leader of the created cycle and no
element in the future will be the leader of this cycle.
Hence, till the end of the execution of the whole algorithm, the cycle will remain 
unchanged.

\begin{theorem}
  There exists~an~algorithm for inverting permutation on
  $n$ elements in $\Oh(n \log n)$ time using $\Oh(\log^2 n)$ additional bits of
  space.
\end{theorem}
\begin{proof}
  The algorithm calls $\Process(i)$ from Algorithm~\ref{alg:process_with_best_inverse} for all $i=1,2,...,n$.
  Until reaching the leader, it operates on the original cycle and the total time is $\Oh(n \log n)$ as in 
  Theorem~\ref{thm:fich}.
  After inverting the cycle, the time for processing an element on the path does not exceed the time for processing this element on the inverted cycle, so the total time spent by the algorithm after inverting the cycle is upper bounded by processing the whole inverted cycle, which is also $\Oh(n\log n)$
  by the same argument.
  
  As we use $k=2 \lceil \log n \rceil$ types of nulls to~store extended ranks of
  leaders, we bound the additional number of bits using Lemma~\ref{le:nulls} for $c = 1$ and $k=2 \lceil \log n \rceil$.
  Recursion in $\Next$ has depth $O(\log n)$ and each level of~recursion requires only $O(\log n)$
  space.
  Thus, in total we use $O(\log^2 n)$ additional bits of space.
\end{proof}

\section{Tortoise and hare algorithm}\label{se:tortoise-and-hare}
\begin{lemma}
  Given an element $i$ on a sigma, it is possible to find the intersection of
  the sigma or report that $i$ is on the loop in $\Oh(\ell+p)$ time, where
  $\ell$ and $p$ are lengths respectively of loop and path of the sigma.
  \label{lemma:intersection-of-sigma}
\end{lemma}
\begin{proof}
  Consider a sigma $\sigma$.
  We can use the classic Floyd's cycle-finding algorithm called tortoise and
  hare.
  In this algorithm we maintain two pointers, initially set to $i$. In each
  iteration both pointers proceed along the sigma, where tortoise does one step
  per iteration and hare does two.
  Eventually, after at most $\ell+p$ of steps they will meet in a common point of
  $\sigma$ and the point belongs to the loop of $\sigma$.

  Then we make the tortoise move from the meeting point to traverse the whole
  loop, obtaining the length $\ell$ of the loop of sigma.
  If the tortoise visits element $i$, then we return that $i$ belongs to the
  loop.
  Otherwise, we can use two tortoises which move one step per iteration from
  the starting point $i$. One of them first does $\ell$ steps and then they
  simultaneously move by one step.
  Clearly, they will meet in the intersection of the sigma.
  See Algorithm~\ref{alg:tortoise-and-hare}.
  Then the total number of steps made by procedure $\FindIntersection(i)$
  is bounded by $3\ell+2p =\Oh(\ell+p)$.
\end{proof}

  \begin{algorithm}[h]
    \begin{algorithmic}[1]
    \Function{$\FindIntersection$}{$i$}
      \State $k \gets 0$,~ $t \gets i$,~ $h \gets i$
      \Repeat
        \If{$\pi(h)=\perp\textbf{ or }\pi(\pi(h))=\perp$} \textbf{abort}
        \EndIf
        \State $t \gets \pi(t)$,~ $h \gets \pi(\pi(h))$
        \State $k \gets k + 1$
      \Until{$t = h$}
      \State $\ell \gets 0$,~ $t_0 \gets t$
      \Repeat
        \If{$t_0=i$} \Return ``$i$ is on the loop''
        \EndIf
        \State $t_0 \gets \pi(t_0)$
        \State $\ell \gets \ell + 1$
      \Until{$t_0 = t$}
      \State $t_1 \gets i$,~ $t_2 \gets i$
      \For{$j = 1..\ell$}
        \State $t_1 \gets \pi(t_1)$
      \EndFor
      \While{$t_1 \neq t_2$}
        \State $t_1 \gets \pi(t_1)$,~ $t_2 \gets \pi(t_2)$
      \EndWhile
      \State \Return ``$t_1$ is the intersection''
    \EndFunction
    \end{algorithmic}
    \caption{Implementation of tortoise and hare technique to find the
    intersection of sigma.}
    \label{alg:tortoise-and-hare}
  \end{algorithm}
\end{document}